\documentclass[aps,nofootinbib,floatfix,a4paper]{revtex4}

\usepackage{amsmath,amsthm,amsfonts,amssymb,thmtools,times,bbm,graphicx,xcolor}

\usepackage[utf8]{inputenc}

\usepackage{hyperref}

\renewcommand{\Re}{\operatorname{Re}}

\usepackage[font=small,justification=RaggedRight]{caption} %
\usepackage{subfig}

\usepackage{tikz}

\tikzset{widthone/.style={draw, minimum width=0.6cm, fill=white, minimum height=16pt, inner sep=-10pt}}

\newtheorem{theorem}{Theorem}
\newtheorem{example}[theorem]{Example}
\newtheorem{proposition}[theorem]{Proposition}
\newtheorem{lemma}[theorem]{Lemma}

\newcommand{\sym}{^\mathrm{sym}}
\newcommand{\sparse}{^\mathrm{sparse}}

\DeclareMathOperator{\diag}{diag}

\def\Spp{{\bf S}}

\def\ZZ{\mathbbm{Z}}

\def\CC{\mathbbm{C}}
\def\FF{\mathbbm{F}}
\def\NN{\mathbbm{N}}

\def\QQ{\mathbbm{Q}}

\def\H{\mathcal{H}}
\def\A{\mathcal{A}}
\def\B{\mathcal{B}}

\def\F{\mathcal{F}}
\def\J{\mathcal{J}}
\def\K{\mathcal{K}}
\def\M{\mathcal{M}}

\def\R{\mathcal{R}}
\def\Id{\mathbbm{1}}

\DeclareMathOperator{\Sp}{Sp}
\DeclareMathOperator{\GSp}{GSp}
\DeclareMathOperator{\ESp}{ESp}
\DeclareMathOperator{\GL}{GL}
\DeclareMathOperator{\SL}{SL}
\DeclareMathOperator{\SO}{SO}

\DeclareMathOperator{\tr}{tr}
\DeclareMathOperator{\Tr}{Tr}

\DeclareMathOperator{\range}{range}
\DeclareMathOperator{\Span}{span}

\DeclareMathOperator{\Sym}{Sym}

\newcommand{\twisted}{\,\tilde\star\,}
\renewcommand{\vec}{\mathbf}

\begin{document}

\title{Thirty-six officers, artisanally entangled}

\author{David Gross and Paulina Goedicke}

\affiliation{
	Institute for Theoretical Physics, University of Cologne, Germany
}

\date{\today}

\begin{abstract} 
	A \emph{perfect tensor} of order $d$ is a state of four $d$-level systems that is maximally entangled under any bipartition.
	These objects have attracted considerable attention in quantum information and many-body theory.
	Perfect tensors generalize the combinatorial notion of \emph{orthogonal Latin squares} (OLS).
	Deciding whether OLS of a given order exist has historically been a difficult problem.
	The case $d=6$ proved particularly thorny, and was popularized by 
	Leonhard Euler in terms of a putative constellation of ``36 officers''. 
	It took more than a century to show that Euler's puzzle has no solution. 
	After yet another century, its quantum generalization 
	was resolved in the affirmative: 36 \emph{entangled} officers can be suitably arranged.
	However, the construction 
	and verification 
	of known instances relies on elaborate computer codes.
	In this paper, we present the first human-made order-$6$ perfect tensors. 
	We decompose the Hilbert space $(\CC^6)^{\otimes 2}$ of two quhexes
	into the direct sum $(\CC^3)^{\otimes 2}\oplus(\CC^3)^{\otimes 3}$ comprising superpositions of two-qutrit and three-qutrit states.
	Perfect tensors arise when certain Clifford unitaries are applied separately to the two sectors.  
	Technically, our construction realizes solutions to the \emph{perfect functions} ansatz recently proposed by Rather.
	Generalizing an observation of Bruzda and \.Zyczkowski, we show that any solution of this kind gives rise to a two-unitary complex Hadamard matrix, of which we construct infinite families.
	Finally, we sketch a formulation of the theory of perfect tensors in terms of quasi-orthogonal decompositions of matrix algebras.
\end{abstract}

\maketitle

\section{Introduction}

\subsection{Motivation}

The systematic study of perfect tensors goes back at least to Ref.~\cite{Scott2004Multipartite}, where the concept appeared as a special case of what the author called \emph{$m$-uniform states} and was linked to the performance of quantum error correction codes.
In the later literature, perfect tensors were primarily discussed under the label of \emph{absolutely maximally entangled} (AME) four-partite states.
The name seems to originate with Ref.~\cite{Helwig2012Absolute}, which used such states to analyze quantum secret sharing and multipartite teleportation protocols.
The dimensions and number of systems 
(beyond four) 
for which AME states exist is the subject of ongoing research Ref.~\cite{Huber2018Bounds, huberAME, 
raissi2020constructions,
kwon2025continuous}.
The concept has also been studied in the field of ``quantum combinatorics'' 
\cite{Zauner1999Quantum}
under various labels, most directly as 
quantum (orthogonal) Latin squares~\cite{Musto2015Quantum, Musto2019Orthogonality, Goyeneche2018Entanglement, clarisse2005entangling,Zyczkowski2021Genuinely}.
More recently, perfect tensors have become the basis for the construction of tensor network models studied in holography \cite{pastawski2015holographic,latorre2015holographic}. 
It was in this context that they received their ``perfect'' moniker \cite{pastawski2015holographic}.
In these models, perfect tensors strengthen 
\cite{Rather2021From}
the related concept of \emph{dual unitaries} \cite{Bertini2019Exact}, which has recently attracted considerable interest in quantum many-body theory
(see, e.g., Refs.~\cite{Claeys2021Ergodic,
Pappalardi2024Quantum,
Rather2021From,
Yu2024Hierarchical}
and references therein). 

Independent of applications, these puzzles might be seen as possessing an intrinsic allure.

\subsection{Definition}
\label{sec:definition}

Let $\{|ij\rangle = |i\rangle\otimes |j\rangle\}$ be the standard product basis of $\CC^d\otimes \CC^d$.
Given a linear map $U$ on $\CC^d\otimes\CC^d$, 
define its \emph{partial transpose} $U^\Gamma$ and \emph{realignment} $U^R$ as the linear maps with matrix elements
\begin{align*}
	\langle ij|U^\Gamma|kl\rangle = \langle il|U|kj\rangle, %
	\qquad
	\langle ij|U^R|kl\rangle = \langle ik|U|jl\rangle.
\end{align*}
Put differently,
the partial transpose operation $U\mapsto U^\Gamma$ is the linear extension of the map defined on product operators as $A\otimes B \mapsto A \otimes (B^t)$.
The realignment is 
\begin{align}\label{eqn:realignment_flip}
	U^R = (U \, \FF)^\Gamma, \qquad\text{in terms of the \emph{flip operator}}\qquad \FF: |ij\rangle \mapsto |ji\rangle. 
\end{align}

Now assume that $U$ is unitary.
Then $U$ is called \emph{dual-unitary} \cite{Bertini2019Exact}, if $U^R$ is unitary as well.
It is
\emph{$\Gamma$-dual unitary} 
\cite{Rather2021From}
if the same is true for $U^\Gamma$.
Finally, $U$ is
\emph{two-unitary} if $U$ is both dual and $\Gamma$-dual unitary 
\cite{Goyeneche2015Absolutely}.
A degree-four tensor $T$ is called \emph{perfect} if $T_{ijkl}=\langle i j|U|kl\rangle$ for some two-unitary $U$
\cite{pastawski2015holographic}.
The two notions are therefore trivially equivalent.
We have used ``perfect tensor'' in the abstract because it sounds more pleasant than the anodyne ``two-unitary''.
However, we will stick to ``two-unitary'' from now, because unitaries are more natural for our purposes than tensors.

A \emph{Latin square of order $d$} is 
a $d\times d$-table $K_{ij}$ with entries in $[d]=\{0, \dots, d-1\}$ such that every row and every column contains each element of $[d]$ exactly once.
Two Latin squares $K, L$ are \emph{orthogonal} if every element of $[d]\times [d]$ occurs exactly once among the pairs $(K_{ij}, L_{ij})$. %
It is easy to see  %
that a pair of orthogonal Latin squares defines a two-unitary $U$ via
\begin{align}\label{eqn:ols_to_unitary}
	U = \sum_{ij} |K_{ij} L_{ij}\rangle\langle ij|.
\end{align}
It is in this sense that two-unitaries generalize the notion of orthogonal Latin squares.

Regrettably, the geometric elegance of the field this paper pertains to is not reflected in its terminology.
The notions \emph{perfect tensors};
\emph{absolutely maximally entangled states}
of four systems (AME(4,$d$))
and 
\emph{$2$-uniform states} of four systems;
\emph{two-unitaries};
and
\emph{quantum orthogonal Latin squares}
all refer to the same mathematical structure.
\emph{Two}-unitaries
are defined by \emph{three} unitarity conditions, and must not be confused with \emph{dual} unitaries.
The only consolation is that the classical situation is only somewhat better.
\emph{Orthogonal Latin squares} and
\emph{Graeco-Latin squares} are the same thing,
and are sometimes (redundantly) referred to as a pair of \emph{mutually orthogonal Latin squares}.

\subsection{Constructions}

\subsubsection{Finite field arithmetic}

There is a simple construction of OLS for prime-power $d$ not equal to $2$.
To this end, identify the elements of $[d]$ with those of the finite field $\FF_d$.
Choose some $0,1\neq \alpha \in \FF_d$ and set
\begin{align}\label{eqn:linear_ols}
	K_{ij} = i+j,
	\quad
	L_{ij} = i + \alpha j
	\qquad
	\text{or, in matrix notation,}
	\qquad
	\begin{pmatrix}
		K\\
		L
	\end{pmatrix}
	=
	\begin{pmatrix}
		1 & 1 \\
		1 & \alpha
	\end{pmatrix}
	\begin{pmatrix}
		i\\
		j
	\end{pmatrix}.
\end{align}
From the representation on the left, it is obvious that $K, L$ are Latin squares.
Because the determinant of the matrix on the right is $\alpha-1\neq 0$, it follows that $(K,L)$ attains every value in $\FF_d^2$, thus establishing orthogonality.

\subsubsection{Products}
\label{sec:products}

If $d=d_1 d_2$ is a product, then the Hilbert space $\CC^d$ is itself isomorphic to a tensor product $\CC^{d_1}\boxtimes \CC^{d_2}$.
Here, we have introduced the notation ``$\boxtimes$'' for internal tensor products, 
in order to distinguish it from the product structure 
\begin{align*}
	\CC^d\otimes \CC^d \simeq (\CC^{d_1} \boxtimes \CC^{d_2}) \otimes (\CC^{d_1}\boxtimes \CC^{d_2})
\end{align*}
with respect to which the notion of two-unitarity is defined.

It then holds that if $U_1, U_2$ are two-unitaries of order $d_1, d_2$ respectively, then their tensor product is a two-unitary of order $d_1 d_2$. 
This follows directly from the fact that the partial transpose and the realignment act separately on tensor factors
\begin{align*}
	(A\boxtimes B)^\Gamma 
	= 
	A^\Gamma \boxtimes B^\Gamma,
	\qquad
	(A\boxtimes B)^R
	= 
	A^R \boxtimes B^R.
\end{align*}

\subsubsection{Order six}

Taken together, the two constructions show the existence of OLS, and hence of two-unitaries, for every order $d$ that is not congruent to $2$ modulo $4$.
For $d=2$, no two-unitaries, and hence no OLS, exist.
This follows from Ref.~\cite{higuchi2000entangled} (see also Ref.~\cite{Goyeneche2015Absolutely} for background).

Leonard Euler conjectured in 1782 that for any order $d\equiv 2\,(\operatorname{mod} 4)$, no OLS exist \cite{colbourn2001mutually}.
In 1900, Tarry showed that there is indeed no solution for $d=6$.
But in the 1960s, the remainder of Euler's conjecture was disproved by Bose, Shrikhande, and Parker:
OLS \emph{can} be constructed for any $d\neq 2,6$
\cite{colbourn2001mutually}.

This leaves the existence of two-unitaries of order $d=6$ as the only open case.
The question gained brief notoriety, being included in a prominent list of open problems in quantum information \cite{horodecki2022five}.

An elaborate computer search 
\cite{rather2022thirty, 
Burchardt2022Symmetry,
zyczkowski2023understanding}
resolved the question in the affirmative shortly after.
Further computer-found solutions were later reported in Refs.~\cite{Rather2024Construction, Bruzda2025Twounitary}.
All these instances are exact, in the sense that each matrix element of the two-unitary is given as an algebraic number.
The solutions display what seems like tantalizing symmetries.
However, their structure has remained unexplained so far, and a manual construction or verification does not seem to be reasonably possible.

\section{Summary of results}

Here, we summarize our results.
Precise definitions and proofs are provided in later sections.

\subsection{Doubly perfect functions}
\label{sec:summary6}

Our constructions build on a mild generalization of an ansatz due to Rather~\cite{Rather2024Construction},
which in turn generalizes methods that have been developed in applied math and harmonic analysis \cite{calabro1967synthesis,bjorck1990functions,fuhr2015biunimodular}.
For $d,n\in\NN$, consider $V:=\ZZ_d^{2n}$, which we will think of as a discrete phase space.
With every vector $\vec a\in V$, associate the \emph{Weyl-Heisenberg operator} (or \emph{generalized Pauli operator}) $w(\vec a)$, which acts on the Hilbert space 
$\H=(\CC^d)^{\otimes n}$.
Let
$|\Phi\rangle=d^{-n/2}\sum_{\vec q\in\ZZ_d^n} |\vec q, \vec q\rangle$ 
be the standard maximally entangled state in 
$\H\otimes\H$.
Then the set $\{|\Phi_{\vec a}\rangle = \big(w(\vec a)\otimes \Id\big)|\Phi\rangle \}_{\vec a \in V}$ forms the \emph{Weyl-Heisenberg ortho-normal basis}
of
$\H\otimes\H$,
consisting of maximally entangled stabilizer states.

Rather worked out the conditions on functions $\lambda: V \to \CC$ so that the Weyl-Heisenberg-diagonal operator
\begin{align}\label{eqn:perfect_wh}
	U_\lambda
	=
	\sum_{\vec a\in V} \lambda(\vec a) |\Phi_{\vec a}\rangle\langle\Phi_{\vec a}|
\end{align}
is two-unitary
(conditions for dual-unitarity were obtained before, see references in \cite{Rather2024Construction}, in particular Refs.~\cite{tyson2003operator, Yu2024Hierarchical}).
To state them, introduce the \emph{standard symplectic form} 
\begin{align}\label{eqn:symplectic}
	[\vec a, \vec b] = \vec a^t J \vec b,
	\qquad
	J=
	\begin{pmatrix}
		0_{n\times n} & \Id_{n\times n} \\
		-\Id_{n\times n} & 0_{n\times n} 
	\end{pmatrix}
\end{align}
on $V$.
Define the \emph{cross-correlation} and the \emph{twisted cross-correlation} of two functions $f, g: V\to\CC$ as 
\begin{align}\label{eqn:correlation_functions}
			(f\star g)(\vec a) %
			=
		\sum_{\vec b\in V} %
			\bar f(\vec b) g(\vec a + \vec b),
			\qquad
			(f\twisted g)(\vec a)
			=
		\sum_{\vec b\in V} %
			\bar f(\vec b) g(\vec a + \vec b) \omega_d^{[\vec a, \vec b]},
			\qquad
			\omega_d=e^{i\frac{2\pi}d}.
\end{align}
Then a short calculation (see Sec.~\ref{sec:auto_proof}) shows that
\begin{align}\label{eqn:autocorr_conditions}
	\begin{split}
		U_\lambda\text{ is unitary } &\Leftrightarrow\; |\lambda|=1, \\
		U_\lambda\text{ is dual-unitary } &\Leftrightarrow\; \lambda\star\lambda = d^{2n} \, \delta, \\ 
		U_\lambda\text{ is $\Gamma$-dual unitary } &\Leftrightarrow\; \lambda\twisted\lambda = d^{2n} \, \delta,
	\end{split}
\end{align}
where $\delta$ is the usual delta function on $V$.

Thus:
To construct a perfect tensor, one has to find a unimodular function on $V$ with no standard or twisted auto-correlations.

It would therefore be just perfect if we could call a solution to the above three equations a ``perfect function''.
However, the nominative curse of the field strikes again:
In classical signal analysis, that name is already taken, and refers to solutions to the first two equations alone.
Rather has suggested using \emph{perfectly perfect} for perfect functions with no twisted auto-correlations.
Here, we opt for \emph{doubly perfect} instead.

\subsubsection{Artisanal doubly perfect function of order $6$}
\label{sec:artisanal}

Reference~\cite{Rather2024Construction} reported that a computer search found doubly perfect functions, which, moreover, take values that are powers of $\omega_6$, and in one case powers of $\omega_3$.
Our main result is the description of \emph{artisanal} doubly perfect functions,
i.e.\ ones that can be constructed, verified, and classified without computer assistance.

To reap the benefits of finite-field arithmetic, we use the numerical coincidence
$36=3^3 + 3^2$
to decompose the phase space $\ZZ_6^2$ into a disjoint union of two vector spaces over $\ZZ_3$
(in contrast, the original order-$6$ perfect tensor is more naturally represented in terms of the decomposition $36=9\times 4$ \cite{zyczkowski2023understanding}).

Starting point is the Chinese remainder isomorphism 
\begin{align*}
	\ZZ_{6} \simeq \ZZ_3 \times \ZZ_2,
	\qquad
	a \mapsto (a\,\operatorname{mod}3, a\,\operatorname{mod}2) =: (k,x).
\end{align*}
Given an element $x\in\ZZ_2=\{0,1\}$, let $\hat x$ be the number obtained by lifting it to $\ZZ_3=\{0,1,2\}$ in the natural way. 
We will work with the invertible map
\begin{align}\label{eqn:decompose}
	\ZZ_6^2
	\to
	\ZZ_3^2 \cup \ZZ_3^3,
	\qquad	
	\vec a
	=
	\begin{pmatrix}
		a_1\\
		a_2
	\end{pmatrix}
	\simeq
	\begin{pmatrix}
		(k,x)\\
		(l,y)\\
	\end{pmatrix}
	&\mapsto
	\left\{
		\begin{array}{ll}
			k,l & (x,y) = (1,1), \\
			k,l,m\quad & (x,y)\neq (1,1)
		\end{array}
	\right.,
	\qquad
	m= \hat x-\hat y .
\end{align}
In words:
First apply the Chinese remainder isomorphism to each component of $\vec a\in\ZZ_6^2$,
resulting in two elements $(k,x)$ and $(l,y)$ of $\ZZ_3\times \ZZ_2$.
Then use $(k,l)$ as the first two components of a $\ZZ_3$-valued vector, and 
if $(x,y)\neq (1,1)$, add a third component equal to $m=\hat x-\hat y$.

In odd dimensions, it is easy to construct doubly perfect functions by taking complex exponentials of quadratic forms.
We aim to mirror this strategy as closely as possible, by allowing for one quadratic form on each of the two components.
Our main theorem classifies the doubly perfect functions that can arise in such a framework.

\begin{restatable}{theorem}{mainthm}\label{thm:main}
	Consider the set of functions $\lambda: V\to \CC$
	defined by two quadratic forms,
	$P(k,l)$ on $\ZZ_3^2$
	and $Q(k,l,m)$ on $\ZZ_3^3$,
	via
	\begin{align*}%
		\lambda(\vec a)
		=
		\omega_3^{\phi(\vec a)},
		\qquad
		\phi(\vec a)
		=
		\left\{
			\begin{array}{ll}
				P(k,l) \quad& (x,y)=(1,1) \\
				P(k,l) + Q(k,l,m) \quad& (x,y)\neq (1,1) \\
			\end{array}
		\right..
	\end{align*}
	Under the action of $\GL(\ZZ_3^2)$ on $(k,l)$, there are exactly two orbits of doubly perfect functions in this class, with representatives 
	\begin{align}
		\lambda\sym:& &P&=k^2+l^2, & Q& = -(k+l+m)^2,  \label{eqn:symmetric} \\
		\lambda\sparse:& &P&=k^2+l^2, & Q&=\phantom{-}(l+m)^2. \label{eqn:sparse}
	\end{align}
\end{restatable}

Honoring the custom of naming order-six perfect tensors \cite{rather2022thirty},
we call the two reference functions the \emph{symmetric artisanal} solution and the \emph{sparse artisanal} solution, respectively.
The names allude to the fact that for the symmetric one, both quadratic forms are invariant under a permutation of their variables, while the sparse function involves fewer non-zero terms.
The latter property makes $\lambda\sparse$ more conducive to calculations, so we will mainly work with this solution in what follows.

On the Hilbert space level, the partitioning of phase space in Eq.~(\ref{eqn:decompose}) translates to the familiar singlet-triplet decomposition 
\begin{align*}
	\CC^2\otimes \CC^2 \simeq \wedge^2(\CC^2) \oplus \Sym^2(\CC^2) \simeq \CC \oplus \CC^3.
\end{align*}
A two-unitary $U_{\lambda}$ associated with such a doubly perfect function
is a direct sum $U_{2}\oplus U_3$ of two Cliffords, with 
$U_2$ an order-three two-unitary on 
$\CC^3\otimes\CC^3\otimes\wedge^2(\CC^2) \simeq (\CC^3)^{\otimes 2}$,
and 
$U_3$ acting on 
$\CC^3\otimes\CC^3\otimes\Sym^2(\CC^2) \simeq (\CC^3)^{\otimes 3}$.

The action of $\GL(\ZZ_3^2)$ on phase space points corresponds to a conjugation of the two-unitaries by local Cliffords and possibly the flip operator. 
Therefore, all two-unitaries associated with the same orbit are manifestly equivalent.
Conversely, $U_{\lambda\sym}$ and $U_{\lambda\sparse}$ have distinct spectrum, so the two orbits describe unitarily inequivalent solutions. %
There are $24$ solutions in each orbit.

\subsubsection{Hadamard two-unitaries}

Bruzda and {\.Z}yczkowski observed in Ref.~\cite{Bruzda2025Twounitary} that there exist two Clifford operations $K, L$ with the remarkable property that, for each of the three doubly perfect functions $\Lambda_i$ reported in Ref.~\cite{Rather2024Construction}, it holds that  $K U_{\Lambda_i} L$ is a two-unitary that is also a complex Hadamard matrix (up to an irrelevant normalization factor).

We prove that this is a general feature of doubly perfect sequences.

\begin{restatable}{theorem}{hadamard}\label{thm:hadamard}
	Given a function $\lambda: V \to \CC$, the complex Hadamard matrices $G$ and $H$ with entries
	\begin{align}\label{eqn:hadamard}
		G_{\vec a, \vec b} = \lambda(\vec a - \vec b) \omega^{[\vec a, \vec b]},
		\quad
		H_{\vec a, \vec b} = \omega^{\vec a_1^t \vec a_2 } \lambda(\vec a - \vec b) \omega^{- \vec b_1^t \vec b_2},
		\qquad
		\vec a = 
		\begin{pmatrix} 
			\vec a_1\\
			\vec a_2
		\end{pmatrix},
		\,
		\vec b = 
		\begin{pmatrix} 
			\vec b_1\\
			\vec b_2
		\end{pmatrix}
		\in \ZZ_d^n \oplus \ZZ_d^n.
	\end{align}
	are proportional to two-unitaries if and only if $\lambda$ is doubly perfect. %
\end{restatable}

One could argue that Eq.~(\ref{eqn:hadamard}) is a more elementary way of establishing the link between doubly perfect functions and two-unitaries, compared to Eq.~(\ref{eqn:perfect_wh}).

To construct examples of two-unitary complex Hadamard matrices, we give explicit doubly perfect sequences for any order $d$, 
unless $d$ is of the form $d=2d_1$, where $d_1$ is neither divisible by $2$ nor by $3$ (the first cases not covered are $d=10, 14, 22$).

\subsubsection{Symmetries}

In Sec.~\ref{sec:symmetries}, we list a number of symmetries that act on the space of doubly perfect functions.

A possibly unexpected element of the symmetry group is the Fourier transform (FT)
\begin{align*}
	\big(\mathcal{F}\lambda)(\vec a)
		&=
		\frac1{d^n}
		\sum_{\vec b\in V}
		\omega_d^{-[\vec a, \vec b]} \lambda(\vec b).
\end{align*}
While it is an immediate consequence of the convolution theorem that the FT of a perfect function is again perfect, it may be less clear that the FT preserves the space of functions with no twisted auto-correlations.
We link this fact to the well-known covariance properties of the \emph{characteristic function} (in the sense of quantum phase spaces).

We also exhibit an explicit symmetry operation which maps 
Rather's $\Lambda_3$ solution
\cite{Rather2024Construction} 
to the symmetric artisanal one.

\subsubsection{Algebraic formulation}

So far, we have approached the theory by considering concrete unitaries, functions, or tensors.
However, in Sec.~\ref{sec:algebra}, we argue that these objects only ``provide coordinates'' for the essential mathematical structure, which is a certain quasi-orthogonal decomposition of matrix algebras.
See Refs.~\cite{Wocjan2005Mutually,
ohno2008quasi, 
Ohno2007QuasiOrthogonal,
gross2012index,
Petz2009Complementarity,
petz2010algebraic,
weiner2010quasi,
appleby2011lie,
weiner2013gap,
reutter2016biunitary,
Nietert2020Rigidity}
for a similar approach to other problems in quantum information theory,
and
Refs.~\cite{pimsner1986entropy,
popa1983orthogonal,
watatani1994latin,
sano1994angles}
as well as Ref.~\cite[Chapter~5]{jones1997introduction}
for related works in the theory of operator algebras.

Let $\M_d$ be the algebra of $d\times d$ matrices.
Then $\M_d\otimes \M_d$ %
contains the two subalgebras
\begin{align*}
	\mathcal{L} = \M_d \otimes \Id, 
	\qquad
	\mathcal{R} = \Id \otimes \M_d.
\end{align*}
Physically, these can be interpreted as local observable algebras of a ``left'' and a ``right'' subsystem.
The Heisenberg picture action of a unitary $U\in \M_{d}\otimes \M_d$ defines an automorphism $X \mapsto U X U^\dagger$ on $\M_d\otimes \M_d$.
Two-unitaries are distinguished by the property that $U\mathcal{L}U^\dagger$ and $U\mathcal{R}U^\dagger$ are  ``maximally delocalized'' in a sense to be defined now.

Given a matrix algebra $\A$, let $\mathcal{A}_0=\{ X \in \A \,|\, \tr X = 0\}$ be the subspace of trace-free elements.
Two subalgebras $\mathcal{A}, \mathcal{B}$ are \emph{quasi-orthogonal} 
\cite{Ohno2007QuasiOrthogonal, ohno2008quasi, weiner2010quasi}
(or \emph{complementary} \cite{Petz2007Complementary, petz2010algebraic})
if $\A_0, \B_0$ are orthogonal w.r.t.\ to the Hilbert-Schmidt inner product.
We then show:

\begin{restatable}{theorem}{algebraic}\label{thm:algebraic}
	The map $U \mapsto U \mathcal{L} U^\dagger =: \A$ defines a one-one correspondence between 
	\begin{itemize}
		\item
			equivalence classes of two-unitaries up to right-multiplication
			by local unitaries,
			$U\mapsto U\,(V_L\otimes V_R)$; and
		\item
			unital subalgebras $\A\subset \M_{d}\otimes \M_d$ that are isomorphic to $\M_d$ and quasi-orthogonal to both local observable algebras.
	\end{itemize}
\end{restatable}

The problem of constructing two-unitaries is thereby reduced to the problem of finding suitable quasi-orthogonal decompositions of matrix algebras.

One could hope that the latter problem is easier than the former.
By way of analogy,
quantum codes are usually specified in terms of a generator tableaux for their stabilizer group.
This is an algebraic description:
The commutant of the stabilizer group is a tensor product of 
an Abelian algebra
(the error syndromes)
and a full matrix algebra
(the observables on the encoded system).
In contrast, it is much less economical to specify the code by giving an explicit encoding operation.
Such an encoder would contain superfluous information:
It would specify not just where the encoded quantum information is stored,
but also how exactly it sits within the encoded space.

At the beginning of this project,
it was our hope that an algebra $\A$ as in Thm.~\ref{thm:algebraic} would not be too difficult to guess.
Unfortunately, this hope has not yet come to pass.
For our concrete calculations, we find it frequently easier to use the explicit form of $U_\lambda$.
However, in Sec.~\ref{sec:uniqueness}, we do use the algebraic approach to derive necessary conditions that ultimately lead to the classification in the main theorem.

\subsection{Outline}

We review a number of concepts related to discrete phase spaces 
(e.g.\ the Weyl-Heisenberg and Clifford group) 
in Sec.~\ref{sec:phase_space}.
The main result is discussed in Sec.~\ref{sec:handmade}.
In Sec.~\ref{sec:symmetries}, we enumerate a number of symmetries of the set of doubly perfect functions.
The connection to complex Hadamard matrices is presented in Sec.~\ref{sec:hadamard}.
We then describe the algebraic approach in Sec.~\ref{sec:algebra}.
Finally, the full symmetry classification is proved in Sec.~\ref{sec:uniqueness}.

\section{Phase-space methods}
\label{sec:phase_space}

The rich structure of finite vector spaces and linear maps made the construction of orthogonal Latin squares in Eq.~(\ref{eqn:linear_ols}) almost trivial.
In many ways, the quantum analogue of a linear structure is given by the Weyl-Heisenberg operators,
and the analogue of linear maps are Clifford operations.
In this section, we briefly recall the basics, based on the presentation in Refs.~\cite{gross2021schur, gross2006hudson}; see also Refs.~\cite{gottesman1998heisenberg,NielsenChuang2011}.

\subsubsection{Weyl-Heisenberg group}

Given a dimension $d$,
label the standard basis $\{|q\rangle\}$ of $\CC^d$ by representatives $x\in\{0, \dots, d-1\}$ of $\ZZ_d=\ZZ/(d\ZZ)$.
With $\omega_d=e^{i\frac{2\pi}d}$ (or just $\omega$, if $d$ is clear from context), define operators
\begin{align*}
	X: |x\rangle \mapsto |x+1\rangle,
	\qquad
	Z: |x\rangle \mapsto \omega_d^{x}|x\rangle
\end{align*}
generalizing the qubit Pauli matrices.

We next describe the composition law of the group generated by $X$ and $Z$ explicitly. 
In order to treat even and odd dimensions in a uniform way, it turns out to be helpful to introduce a further phase factor, $\tau_d=(-1)^d e^{i\pi/d}$, which fulfills $\tau_d^2=\omega_d$. 
If $d$ is odd, then $\tau_d$ is a $d$-th root of unity (and hence equal to $\omega_d^{2^{-1}}$, with $2^{-1}$ the multiplicative inverse of $2$ modulo $d$).
If $d$ is even, then $\tau_d$ is a $2d$-th root of unity.
Associate with every vector $\vec a=(p,q)\in\ZZ^2$ the (single-qudit) \emph{Weyl-Heisenberg} (WH) \emph{operator}
\begin{align*}
	w(\vec a) = \tau_d^{-pq} Z^p X^q.
\end{align*}
The definition extends to $n$ systems:
The WH operator associated with a vector $\vec a=(\vec p, \vec q)\in \ZZ_d^n\oplus\ZZ_d^n=\ZZ_d^{2n}$ is
\begin{align*}
	w(\vec a)
	=
	\bigotimes_{i=1}^n w(p_i, q_i)
\end{align*}
acting on $\big(\CC^d\big)^{\otimes n}$.
A straight-forward calculation verifies the composition law
\begin{align}\label{eqn:wh_composition}
	w(\vec a) w(\vec a') = \tau_d^{[\vec a, \vec a']} w(\vec a+ \vec a'),
\end{align}
in terms of the standard symplectic form (\ref{eqn:symplectic}).

\subsubsection{Clifford group}

The \emph{Clifford group} is the set of unitary automorphisms of the WH group.
The inner automorphisms, i.e.\ the action by conjugation of WH operators, act as the multiplication by a character.
Indeed, from (\ref{eqn:wh_composition}),
\begin{align*}
	w(\vec a) \, w(\vec b) \, w(\vec a)^\dagger = \omega^{[\vec a, \vec b]} \, w(\vec b).
\end{align*}
It turns out that the quotient of the Clifford group by the inner automorphisms and phase factors is isomorphic to the symplectic group $\Sp(\ZZ_d^{2n})$.
More concretely, for every Clifford unitary $U$,
there exists a matrix $S\in\ZZ_d^{2n\times 2n}$ preserving the symplectic form (\ref{eqn:symplectic}), such that
\begin{align*}
	U \, w(\vec a) \, U^\dagger \propto w(S \vec a).
\end{align*}
Conversely, every element of $\Sp(\ZZ_d^{2n})$ is realized this way.

If $d$ is odd, %
the Clifford group even contains a subgroup isomorphic to $\Sp(\ZZ_d^2n)$.
It is known by a number of names, such as the \emph{metaplectic}, \emph{Weil}, or \emph{oscillator} representation.
In physics, unitaries that arise this way are sometimes called \emph{symplectic Clifford unitaries}.
Unfortunately, the literature on this subject is vast but disconnected.
Some starting points are Refs.~\cite{gerardin,folland_harmonic,neuhauser2002explicit,appleby2005symmetric,nebe2006self,gross2021schur}. 

\subsubsection{Concrete symplectic Clifford unitaries}
\label{sec:concrete_cliffords}

Here, we give explicit formulas for some Clifford unitaries that will be used in this paper.

For odd $d$, the formulas agree with the metaplectic representation up to global phases, which we have chosen to simplify the expressions.
(In particular, the unitaries listed below only generate a \emph{projective} representation of the symplectic group -- a minor demerit that causes no issues for our purposes).

	\begin{itemize}
		\item
			For every invertible linear map $G\in \GL(\ZZ_d^n)$, the permutation operator
			\begin{align*}
				U_G^{\GL}: 
				|\vec x\rangle \mapsto |G\,\vec x\rangle
				\qquad\text{ is Clifford, with associated symplectic map } \qquad
				S^{\operatorname{GL}}_G
				=
				\begin{pmatrix}
					G^{-t} & 0 \\
					0 & G
				\end{pmatrix}.
			\end{align*}
			Here, $G^{-t}$ is the transpose of the inverse of $G$, computed modulo $d$.
			Special cases are
			\begin{align*}
				G_{C_1X_2}=
				\begin{pmatrix}
					1 & 0 \\
					1 & 1
				\end{pmatrix},
				\qquad
				G_{OLS}=
				\begin{pmatrix}
					1  & 1 \\
					-1 & 1
				\end{pmatrix}.
			\end{align*}
			The first one implements the controlled-$X$ gate (first system controlling, second system controlled);
			and the second is the quantized version of the linear map (\ref{eqn:linear_ols}) that generates an OLS in odd prime dimensions.
		\item
			The finite Fourier transform (or \emph{Schur matrix}) 
			\begin{align}\label{eqn:schur}
				F:
				|\vec x\rangle
				\mapsto
				\frac1{\sqrt d}
				\sum_{\vec y} 
				\omega^{\vec x \vec y}|\vec y\rangle
				\qquad
				\text{ is Clifford, with associated symplectic map }
				\qquad
				J
				\qquad \text{ as defined in Eq.~(\ref{eqn:symplectic}) }
				.
			\end{align}
			(Here, and in what follows, we write $\vec x \vec y$ for the canonical symmetric form $\vec x^t \vec y$ between elements of $\ZZ_d^n$).
		\item
			If $N$ is a symmetric matrix in $\ZZ^{n \times n}$, 
			then
			\begin{align}\label{eqn:phase_gates}
				U_N^Q: 
				|\vec x\rangle \mapsto \tau^{\vec x N \vec x} |\vec x\rangle
				\qquad\text{ is Clifford, with associated symplectic map } \qquad
				S^{\operatorname{Q}}_N 
				=
				\begin{pmatrix}
					\Id & N \\
					0 & \Id				
				\end{pmatrix}
				\mod d.
			\end{align}
			In particular, choosing $N=e_i e_j^t + e_j e_i^t$ gives rise to the controlled-$Z$ gate between systems $i$ and $j$, and $N=e_i e_i^t$ implements the phase gate on system $i$.
			Conversely, a diagonal unitary is Clifford if and only if it is of the form
			\begin{align}\label{eqn:diagonal_clifford}
				D = e^{i\phi}\,U^Q_N w(\vec p\oplus 0)
				\quad\Leftrightarrow\quad
				D|\vec x\rangle = e^{i\phi} \tau^{2\vec p \vec x+\vec x N \vec x}.
			\end{align}
			for some global phase $\phi$, vector $\vec p \in \ZZ_d^{n}$, and a symmetric matrix $N\in\ZZ^n$.
\end{itemize}

\subsubsection{The extended Clifford group}

The construction of the WH group depends implicitly on the 
choice of a phase factor $\tau_d$. 
The substitution $\tau_d\mapsto \tau_d^m$ for some power $m\in\ZZ$ that is co-prime to the order of $\tau$ gives rise to a faithful, but unitarily inequivalent representation of the WH group.
It shares all algebraic properties of the defining representation in the following sense:
The matrix elements of the WH operators are elements of the cyclotomic field $\mathbb{Q}[\tau]$, and the substitution amounts to the application of an element of its Galois group \cite{Appleby2015GaloisUM, obst2024wigner}. 
(We remark that in the continuous-variable analogue of the theory \cite{folland_harmonic}, the number ``$m$'' is related to the mass parameter of a projective representation of the Galileo group).
In Sec.~\ref{sec:characteristic}, we will consider the interplay of WH representations defined for different dimensions and values of $m$. 
It will then be necessary to make the dependence of the WH operators on these parameters explicit,
which we will do by adding a superscript, as in $w^{(d,m)}(\vec a)$.

An alternative approach is to work with one fixed realization of the WH group,
and extend the Clifford group to also encompass semi-linear maps that include the
Galois-automorphisms of $\mathbb{Q}[\tau]/\mathbb{Q}$ (see Refs.~\cite{Appleby2015GaloisUM,appleby2009properties} for early discussions, and Ref.~\cite{obst2024wigner} for the case $n>1$).

Concretely, if $k$ is an element of the multiplicative group $\ZZ_{d}$ (odd dimension)
or of $\ZZ_{2d}$ (even dimension), 
then there exists a Galois automorphism sending $\tau_d$ to $\tau_d^k$.
It is an element of the extended Clifford group, associated with the matrix
\begin{align*}
	S_k^E
	:=
	\begin{pmatrix}
		k\Id  & 0 \\
		0 & \Id
	\end{pmatrix}.
\end{align*}
It is a \emph{symplectic similitude} in the sense that it preserves the symplectic form only up to a scalar multiple:
\begin{align*}
	[S_k^E \vec a, S_k^E \vec b]
	=
	k [\vec a, \vec b].
\end{align*}
Complex conjugation (corresponding to $k=-1$) is singled out among the Galois automorphisms, as it is the only one that is defined for all of $\CC$.
Following Refs.~\cite{Appleby2015GaloisUM,appleby2009properties}, 
we call the group generated by the Clifford group and complex conjugation the \emph{extended Clifford group}, and the group of symplectic similitudes with scaling $k\in \{+1, -1\}$ the \emph{extended symplectic group} $\ESp(V)$.
The group that includes all field automorphisms is the \emph{Galois Clifford group} and the group of all similitudes is denoted by $\GSp(V)$.

Their structure is particularly simple for $n=1$, as in this case, $\Sp(\ZZ_d^{2})=\operatorname{SL}(\ZZ_d^2)$, and $\GSp(\ZZ_d^2)=\GL(\ZZ_d^2)$.

\subsubsection{The Weyl-Heisenberg basis}
\label{sec:wh_basis}

On $\CC^d\otimes \CC^d$, define the maximally entangled states
\begin{align}\label{eqn:whbasis}
	|\Phi\rangle 
	= \frac1{\sqrt d} \sum_{x=0} ^{d-1} |xx\rangle,
	\qquad
	|\Phi_{(p,q)}\rangle
	=
	(\tau^{pq} w(p,q)\otimes\Id)|\Phi\rangle.
\end{align}
The state $|\Phi_{(p,q)}\rangle$ is the ``vectorization'' of $\tau^{pq}d^{-1/2} \, w(\vec a)$. 
This implies that the set $\{ |\Phi_{\vec a} \rangle \}_{\vec a \in \ZZ_d^n}$ forms an ortho-normal basis, sometimes known as the \emph{Weyl-Heisenberg basis}.

The unitary 
\begin{align}\label{eqn:uwh}
	U_{WH}=: |p,q\rangle \mapsto |\Phi_{(p,q)}\rangle
\end{align}
is Clifford.
It can be realized as the Fourier transform on the second system, 
followed by a controlled-$X$ gate, where the first system controls the execution on the second one:
\begin{align}\label{eqn:swh}
	U_{WH}
	=
	U^{\GL}_{C_1X_2}
	(\Id\otimes F)
	\quad
	\Rightarrow
	\quad
	S_{WH}
	:=
	\underbrace{
		\left(
			\begin{array}{cccc}
			 1 & 0 & 0 & 0 \\
			 -1 & 1 & 0 & 0 \\
			 0 & 0 & 1 & 1 \\
			 0 & 0 & 0 & 1 \\
			\end{array}
		\right)
	}_{\text{controlled$_2$-$X_1$}}
	\underbrace{
		\left(
			\begin{array}{cccc}
			 1 & 0 & 0 & 0 \\
			 0 & 0 & 0 & 1 \\
			 0 & 0 & 1 & 0 \\
			 0 & -1 & 0 & 0 \\
			\end{array}
		\right)
	}_{\text{FT on 2nd system}}
	=
	 \left(
	 	\begin{array}{cccc}
	 	 1  & 0  & 0 & 0 \\
	 	 -1 & 0  & 0 & 1 \\
	 	 0  & -1 & 1 & 0 \\
	 	 0  & -1 & 0 & 0 \\
	 	\end{array}
	 \right)
	.
\end{align}

\begin{proof}
	To see that the Clifford unitary $U_{WH}$ associated with 
	$S_{WH}$
	indeed implements the transformation (\ref{eqn:uwh}), start with the stabilizer equations
	\begin{align*}
		w(p_1,0)\otimes w(p_2,0) |p,q\rangle
		=
		\omega^{p_1 p + p_2 q} |q,p\rangle 
		\quad\Rightarrow\quad
		U_{WH}w(p_1,0)\otimes w(p_2,0)U_{WH}^\dagger U_{WH}|p,q\rangle
		=
		\omega^{p_1 p + p_2 q} 
		U_{WH}
		|q,p\rangle .
	\end{align*}
	The conjugation on the l.h.s.\ gives a WH operator with parameter $S_{WH}(p_1,p_1,0,0)^t=(p_1,-p_1,-p_2,-p_2)^t$.
	Using the relation
	\begin{align*}
		(A\otimes B)|\Phi\rangle = (A B^t \otimes \Id) |\Phi\rangle,
	\end{align*}
	this is seen to coincide with the stabilizer equations of the WH basis:
	\begin{align}
		\begin{split}
		\label{eqn:wh_stabilizer}
		w(p_1,-p_2) \otimes w(-p_1, -p_2) |\Phi_{p,q}\rangle
		&=
		\tau^{pq} 
		w(p_1,-p_2) w(p,q) \otimes w(-p_1, -p_2) |\Phi\rangle \\
		&=
		\tau^{pq} 
		w(p_1,-p_2) w(p,q) w(-p_1,p_2)\otimes \Id |\Phi\rangle 
		=
		\omega^{p_1 q + p_2 p}
		|\Phi_{{p,q}}\rangle.
		\end{split}
	\end{align}
\end{proof}

The WH basis for $n>1$ is just the $n$-fold tensor product of the WH basis introduced above.

Remarks:
\begin{itemize}
	\item
		A more canonic approach 
		(used in Sec.~\ref{sec:auto_proof})
		is to model the bi-partite Hilbert space directly by $L(\CC^d)$ with Hilbert-Schmidt inner product, instead of working with $\CC^d\otimes\CC^d$.
		Then 
		the $w(\vec a)$ 
		can be interpreted as state vectors, obviating the use of the non-canonic isomorphism $L(\CC^d)\to \CC^d\otimes \CC^d$ given by $w(\vec a) \mapsto |\Phi_{\vec a}\rangle$.
		The challenge of this approach is the possibility of confusion, due to linear maps now representing both elements in a Hilbert space and operations on a Hilbert space.

	\item
		We are interested in the WH basis mainly for the construction of the unitary $U_\lambda$ given in Eq.~(\ref{eqn:perfect_wh}).
		This operator is clearly not affected by phase changes $|\Phi_{\vec a}\rangle \mapsto e^{i\phi_\vec a} |\Phi_{\vec a}\rangle$.
		The phases in Eq.~(\ref{eqn:whbasis}) have been chosen to make the basis change Clifford.
		At times, we will find it more convenient to work with other phase conventions.
		Natural choices are:
		(1) $|\Phi_{\vec a}\rangle = w(\vec a)\otimes \Id|\Phi\rangle$.
		This corresponds to the approach mentioned in the previous remark, was used in the introduction and will again be used used in Sec.~\ref{sec:auto_proof}. 
		In odd dimensions, it also does give rise to a Clifford basis change.
		(2) For qubits, $\{ |\Phi_{0,0}\rangle, -i |\Phi_{0,1}\rangle, - |\Phi_{1,1}\rangle, -i |\Phi_{1,0}\rangle \}$.
		This version is connected to the representation of the quaternion algebra in $L(\CC^2)$ and makes a Lie algebra isomorphism used in Sec.~\ref{sec:lie} cleaner.
\end{itemize}

\subsubsection{The Chinese Remainder Obstruction}
\label{sec:great_wall}

Recall that the Chinese remainder theorem says that if $d=d_1 \, d_2$ is a product of co-prime numbers, then $x\mapsto (x\, \operatorname{mod} d_1, x \, \operatorname{mod} d_2)$ implements a ring isomorphism of $\ZZ_d\to\ZZ_{d_1}\times \ZZ_{d_2}$.

Applied to basis states, this gives rise to a unitary equivalence $R: \CC^{d}\simeq \CC^{d_1}\otimes \CC^{d_2}$.
Because the definition of the WH group reduces to operations in the ring $\ZZ_d$, it turns out 
\cite{gross2008culs,appleby2012monomial}
that WH operators factorize with respect to this product structure.
More precisely,
\begin{align}\label{eqn:chinese_wh}
	R w^{(d_1d_2,m)}(p,q) R^\dagger 
	\simeq 
	w^{(d_1,\kappa_1m)}(p,q) \otimes w^{(d_2,\kappa_2m)}(p,q)
\end{align}
where $\kappa_i$ is the multiplicative inverse of $d/d_i$, taken modulo $d_i$ for odd $d_i$ and modulo $2d_i$ for even $d_i$.
The relevant case for us is $d_1=3, d_2=2$, where the general formula gives $\kappa_1=\kappa_2=-1$, which amounts to a complex conjugation of the $\tau$ factors:
\begin{align*}
	R w^{(6,1)}(p,q) R^\dagger 
	=
	{ w^{(3,-1)}(p,q)}
	\otimes 
	{w^{(2,-1)}(p,q)}
	=
	\overline{ w^{(3,1)}(p,q)}
	\otimes 
	\overline{w^{(2,1)}(p,q)}
	.
\end{align*}
The result extends from WH operators to all Clifford unitaries \cite{gross2008culs,appleby2012monomial}:
They, too, factorize,
and in the special case $d=3\times2$,
the factors are given by the complex conjugation of the defining representation.

More physically speaking, Clifford operations cannot create entanglement between spaces of co-prime dimension.

This gives rise to an obstruction one might call the \emph{Great Wall}.
No Clifford operation can be two-unitary in dimensions $d$ congruent to $2$ modulo $4$.
That's because any such unitary would factorize, and the qubit factor would be two-unitary on $\CC^2\otimes \CC^2$,  which is impossible.

\section{Hand-made doubly perfect functions of order six}
\label{sec:handmade}

\subsection{Ansatz}
\label{sec:ansatz}

Here, we will motivate the ansatz we have chosen.

As a warm-up, consider the case where $\lambda(\vec a) = \tau^{\vec a N \vec a}$ is a quadratic form defined by a symmetric matrix $N$ for general $d,n$.
Then 
\begin{align}\label{eqn:perfect_quadratic}
	\lambda({\vec a + \vec b}) \bar\lambda({\vec b})
	=
	\tau^{\vec a N \vec a}
	\,
	\omega^{\vec a N \vec b}
	\qquad
	\Rightarrow
	\qquad
	(\lambda \star \lambda)(\vec a)
	\propto \sum_{\vec b} \omega^{\vec a N \vec b},
	\quad
	(\lambda \twisted \lambda)(\vec a)
	\propto 
	\sum_{\vec b} 
	\omega^{\vec a (N+J) \vec b}
	.
\end{align}
If $N\in\ZZ_d^{2n\times 2n}$, both auto-correlations reduce to a character sum, and we get the sought-for delta functions if and only if both $N$ and $N+J$ have trivial kernel.
Specializing to $n=1$ now, because
\begin{align*}
	\det (N+J) 
	= 
	N_{11} N_{22} - (N_{12}-1)(N_{12}+1)
	=
	\det N
	+1,
\end{align*}
this is equivalent to both $\det N$ and $(\det N)+1$ being co-prime to $d$.
For odd $d$, obvious solutions are given by $N=\pm\Id$.
But since two consecutive numbers can't both be odd, there is no solution for even $d$ (compatible with the general argument in Sec.~\ref{sec:great_wall}).

Now turn to $d=6$.
Throughout, we will use that the Chinese remainder isomorphism $\ZZ_6\simeq \ZZ_3 \times \ZZ_2$, realized by the maps
\begin{align}\label{eqn:cr_with_inverse}
	a \mapsto (a\,\operatorname{mod}3, a\,\operatorname{mod}2)
	\qquad\text{with inverse}\qquad
	(k,x) \mapsto 4k + 3 x.
\end{align}
Applying it component-wise to elements of $\ZZ_6^2$ gives an isomorphism $\ZZ_6^2 \to \ZZ_3^2 \times \ZZ_2^2$.

Now comes the key step.
The quadratic form ansatz led to an almost trivial solution in odd dimensions,
but we cannot naively extend it, as relying on modular arithmetic modulo an even number will trigger the Chinese obstruction.
Thus, it would be nice if one could impose a $\ZZ_3$-linear structure on the $\ZZ_2^2$-part.
We have described a natural approach for achieving this as Eq.~(\ref{eqn:decompose}) in the introduction, Sec.~\ref{sec:artisanal}.
Restating it for convenience:
\begin{align*}
	\ZZ_6^2
	\to
	\ZZ_3^2 \cup \ZZ_3^3,
	\qquad	
	\vec a
	=
	\begin{pmatrix}
		a_1\\
		a_2
	\end{pmatrix}
	\simeq
	\begin{pmatrix}
		(k,x)\\
		(l,y)\\
	\end{pmatrix}
	&\mapsto
	\left\{
		\begin{array}{ll}
			k,l & (x,y) = (1,1), \\
			k,l,m\quad & (x,y)\neq (1,1)
		\end{array}
	\right.,
	\qquad
	m= \hat x-\hat y .
\end{align*}

From the discussion at the beginning of this section, there might be reasonable hope that one can extend the simplest $d=3$-solution, $\lambda=\omega_3^{k^2+l^2}$, to $d=6$ by adding a quadratic form involving the new variable $m$ on the $\ZZ_3^3$-component.
Our main result is that this does indeed work.

\mainthm*

The proof is in Sec.~\ref{sec:uniqueness}.
Because the full symmetry classification is quite lengthy, 
we provide an independent argument showing that $\lambda\sparse$ is doubly perfect in Sec.~\ref{sec:proof}.

\subsection{Unitary implementation}
\label{sec:unitaries}

Using the qubit WH basis
\begin{align*}
	|\Phi_{0,0}\rangle = \frac1{\sqrt{2}}(|00\rangle + |11\rangle), 
	\quad
	|\Phi_{0,1}\rangle = \frac1{\sqrt{2}}(|01\rangle + |10\rangle), 
	\quad
	|\Phi_{1,0}\rangle = \frac1{\sqrt{2}}(|00\rangle - |11\rangle), 
	\quad							 
	|\Phi_{1,1}\rangle = \frac1{\sqrt{2}}(|01\rangle - |10\rangle), %
\end{align*}
one sees that the partition $\ZZ_2^2= \{(0,1), (0,0), (1,0)\} \cup \{ (1,1) \}$ gives rise to the triplet-singlet decomposition
\begin{align*}
	\CC^2\otimes \CC^2 \simeq \wedge^2(\CC^2) \oplus \Sym^2(\CC^2) \simeq \CC \oplus \CC^3.
\end{align*}
With these choices, $U_{\lambda}$ 
is a direct sum, $U_{2}\oplus U_3$, with $U_2$ a symplectic Clifford unitary on
$(\CC^3)^{\otimes 2} \simeq \CC^3\otimes\CC^3\otimes\wedge^2(\CC^2)$,
and $U_3$ a symplectic Clifford unitary on
$(\CC^3)^{\otimes 3} \simeq \CC^3\otimes\CC^3\otimes\Sym^2(\CC^2)$.

More explicitly, using the isometry 
\begin{align*}
	V_t&:= \sum_{x,y\neq (1,1)} 
	|\Phi_{x,y}\rangle\langle\hat x - \hat y|:
	\CC^3 \to \wedge^2(\CC^2)
\end{align*}
from $\CC^3$ to the triplet space,
we can write
\begin{align*}
	U_\lambda
	= 
	U_2 \otimes |\Phi_{1,1}\rangle\langle\Phi_{1,1}|
	+
	(\Id^{\otimes 2} \otimes V_t)
	U_3 
	(\Id^{\otimes 2} \otimes V_t)^\dagger
	\in
	U(\CC^3\otimes\CC^3\otimes\CC^2\otimes\CC^2).
\end{align*}
Using the definitions of Sec.~\ref{sec:concrete_cliffords}, the two- and three-qutrit Cliffords implementing the sparse solution are
\begin{align}\label{eqn:cliffords}
	U_2 
	= 
	U_{WH} U^Q_{N_2} U_{WH}^\dagger,
	\quad
	U_3 
	= 
	(U_{WH}\otimes\Id) U^Q_{N_3} (U_{WH}^\dagger\otimes \Id),
	\qquad
	N_2 = 
	\begin{pmatrix}
		1 & 0  \\
		0 & 1  \\
	\end{pmatrix},
	\qquad
	N_3 = 
	\begin{pmatrix}
		1 & 0  & 0 \\
		0 & 2  & 2 \\
		0 & 2  & 1
	\end{pmatrix}.
\end{align}
Circuit realizations are shown in Fig.~\ref{fig:circuit}.

In later sections, we will need the explicit symplectic matrix $S_2$
associated with 
$U_2$. 
From Sec.~\ref{sec:concrete_cliffords},  it is given by
\begin{align}\label{eqn:23cliff}
	S_2
	=
	S_{WH}
	S_{N_2}^Q
	S_{WH}^{-1}
	=
	\left(
	\begin{array}{cccc}
	 1 & 0 & 1 & 2 \\
	 0 & 1 & 2 & 1 \\
	 2 & 2 & 1 & 0 \\
	 2 & 2 & 0 & 1 \\
	\end{array}
	\right)
	.
\end{align}

\begin{figure}

	\begin{center}
		\begin{tikzpicture}[thick, scale = 0.55]
			\draw (-5,1.5) to (-3,1.5); 
			\draw (-1,2) to node[above]{$\mathbb{C}^3$} (1,2); 
			\draw (3,2)  to (5,2); 
			\draw (-1,1) to node[above]{$\mathbb{C}^2$} (1,1);
			\draw (1,1) to (5,1);
			\draw (-5,-0.5)  to (-3,-0.5);
			\draw (-1,0) to node[above]{$\mathbb{C}^2$} (1,0);
			\draw (1,0) to (3,0);
			\draw (3,0) to (5,0);
			\draw (-1,-1) to  node[above]{$\mathbb{C}^3$} (1,-1);
			\draw (3,-1)  to (5,-1);
			\node[] at (-5.5,1.5) {$\mathbb{C}^6$};
			\node[] at (-5.5,-0.5) {$\mathbb{C}^6$};
			\node[] at (2,1.75) {$WH$};
			\node[] at (-2,0.6) {$\mathrm{R}$};
			\node[widthone] at (4, 2) {$P$};
			\node[widthone] at (4, -1) {$P$};
			\draw (-3, 2.5) to (-1, 2.5);
			\draw (-3, -1.5) to (-1, -1.5);
			\draw (-3, 2.5) to (-3, -1.5);
			\draw (-1, 2.5) to (-1, -1.5);
			\draw (1, 2.5) to (3, 2.5);
			\draw (1, -1.5) to (3, -1.5);
			\draw (1, 2.5) to (1, 1.75);
			\draw (3, 2.5) to (3, 1.75);
			\draw (1, -1.5) to (1, -0.75);
			\draw (3, -1.5) to (3, -0.75);
			\draw[dashed] (1, 1.75) to (1, 1.25);
			\draw[dashed] (3, 1.75) to (3, 1.25);
			\draw[dashed] (1, 0.75) to (1, 0.25);
			\draw[dashed] (3, 0.75) to (3, 0.25);
			\draw[dashed] (1, -0.75) to (1, -0.25);
			\draw[dashed] (3, -0.75) to (3, -0.25);
			\node[] at (9,0.5){$\Bigg\} \ \ \mathrm{Sym}^2 (\mathbb{C}^2) \oplus \Lambda^2 (\mathbb{C}^2) \ \ \Bigg\{ $ };
			\end{tikzpicture}
			\quad
			\begin{tikzpicture}[thick, scale = 0.55]
			\draw (3,1.5) to (5,1.5); 
			\draw (-5,2) to  (-3,2); 
			\draw (-1,2)  to  node[above]{$\mathbb{C}^3$} (1,2); 
			\draw (-5,1) to  (-3,1);
			\draw (-3,1) to  (-1,1);
			\draw (-1,1) to node[above]{$\mathbb{C}^2$} (1,1);
			\draw (3,-0.5)  to (5,-0.5);
			\draw (-5,0)  to (-3,0);
			\draw (-3,0) to  (-1,0);
			\draw (-1,0)  to node[above]{$\mathbb{C}^2$} (1,0);
			\draw (-5,-1) to  (-3,-1);
			\draw (-1,-1)  to node[above]{$\mathbb{C}^3$} (1,-1);
			\node[] at (5.5,1.5) {$\mathbb{C}^6$};
			\node[] at (5.5,-0.5) {$\mathbb{C}^6$};
			\node[] at (2.1,0.7) {$\mathrm{R}^\dagger$};
			\node[] at (-2,1.75) {$WH^\dagger$};
			\draw (-3, 2.5) to (-1, 2.5);
			\draw (-3, -1.5) to (-1, -1.5);
			\draw (3, 2.5) to (3, -1.5);
			\draw (1, 2.5) to (1, -1.5);
			\draw (1, 2.5) to (3, 2.5);
			\draw (1, -1.5) to (3, -1.5);
			\draw (-1, 2.5) to (-1, 1.75);
			\draw (-3, 2.5) to (-3, 1.75);
			\draw (-1, -1.5) to (-1, -0.75);
			\draw (-3, -1.5) to (-3, -0.75);
			\draw[dashed] (-1, 1.75) to (-1, 1.25);
			\draw[dashed] (-3, 1.75) to (-3, 1.25);
			\draw[dashed] (-1, 0.75) to (-1, 0.25);
			\draw[dashed] (-3, 0.75) to (-3, 0.25);
			\draw[dashed] (-1, -0.75) to (-1, -0.25);
			\draw[dashed] (-3, -0.75) to (-3, -0.25);
		\end{tikzpicture}
	\end{center}

	\begin{center}
			\subfloat[For $\lambda\sparse$]{%
				\begin{tikzpicture}[thick, scale= 0.75]
					\draw (0,1)   to (4,1)  ; 
					\draw (0,0)   to (4,0);
					\draw (0,-1)  to (4,-1);
					\draw (1,-1)  to (1,0);
					\draw (3,-1)  to (3,0);
					\node[] at (-0.5,1) {$\mathbb{C}^3$};
					\node[] at (-1.8,0) {$\mathrm{Sym}^2(\mathbb{C}^2) \cong \mathbb{C}^3$};
					\node[] at (-0.5, -1) {$\mathbb{C}^3$};
					\node[widthone] at (1, 0) {$X$};
					\node[widthone] at (2, 0) {$P$};
					\node[widthone] at (3, 0) {$X^\dagger$};
					\node  at (3,-1) [circle,fill, inner sep = 1.5 pt]{};
					\node  at (1,-1) [circle,fill, inner sep = 1.5 pt]{};
				\end{tikzpicture}
			}
			\hspace{2cm}
			\subfloat[For $\lambda\sym$]{%
			\begin{tikzpicture}[thick, scale= 0.75]
			\draw (0,1)  to (6,1)  ; 
			\draw (0,0)  to (6,0);
			\draw (0,-1)  to (6,-1);
			\draw (1,0)   to (1,1);
			\draw (2,-1)   to (2,0);
			\draw (4,-1)   to (4,0);
			\draw (5,0)   to (5,1);
			\node[] at (-0.5,1) {$\mathbb{C}^3$};
			\node[] at (-1.8,0) {$\mathrm{Sym}^2(\mathbb{C}^2) \cong \mathbb{C}^3$};
			\node[] at (-0.5, -1) {$\mathbb{C}^3$};
			\node[widthone] at (3, 0) {$P^\dagger$};
			\node[widthone] at (4, 0) {$X^\dagger$};
			\node[widthone] at (5, 0) {$X^\dagger$};
			\node[widthone] at (1, 0) {$X$};
			\node[widthone] at (2, 0) {$X$};
			\node  at (4,-1) [circle,fill, inner sep = 1.5 pt]{};
			\node  at (1,1) [circle,fill, inner sep = 1.5 pt]{};
			\node  at (2,-1) [circle,fill, inner sep = 1.5 pt]{};
			\node  at (5,1) [circle,fill, inner sep = 1.5 pt]{};
			\end{tikzpicture}
			} 
	\end{center}

	\caption{
			\label{fig:circuit}
			Circuit implementation of $U_\lambda$.
			The Chinese Remainder unitary ($R$) splits each copy of $\CC^6$ into $\CC^3\otimes\CC^2$.
			On the qutrit systems, apply $U_{WH}$, mapping the standard to the Weyl-Heisenberg basis.
			The two phase gates $P$ implement an order-3 two-unitary, as detailed in Sec.~\ref{sec:u2}.
			Think of the qubit space
			$\CC^2\otimes\CC^2$ 
			as a direct sum of 
			$\wedge^2(\CC^2)\simeq \CC$, the singlet space,
			and
			$\Sym^2(\CC^2)\simeq \CC^3$, 
			the triplet space.
			No further action is performed on the singlet sector $(\CC^3)^\otimes \otimes \wedge^2(\CC^2)$.
			The two circuits in the bottom row act on the triplet sector,
			$(\CC^3)^\otimes \otimes \Sym^2(\CC^2)$,
			and result in the sparse and the symmetric solution, respectively.
			Both perform a phase gate in an entangled basis, which is realized by conjugating with controlled-$X$ gates.
			The triplet sector circuits commute with the two $P$ gates shown in the top row.
		}
\end{figure}
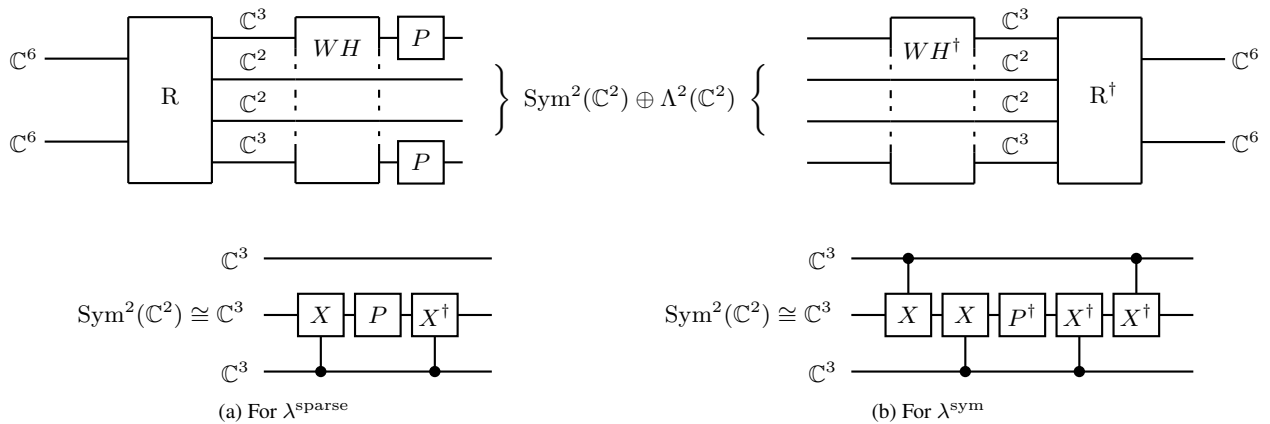

Remark:
\begin{itemize}
	\item
	The fact that the above decomposition of $\CC^2\otimes\CC^2$ is invariant under a swap of the two qubits is inessential, and has been chosen only in order to work with a commonly-used decomposition.
	The property is not used in the proof, and e.g.\ exchanging the role of $00$ and $11$ produces another solution.
\end{itemize}

\subsection{Verification of the sparse artisanal solution}
\label{sec:proof}

In this section, we will verify that Eq.~(\ref{eqn:sparse}) is indeed doubly perfect, by directly checking the defining auto-correlation equations.
All results obtained here are implied by the full symmetry classification in Sec.~\ref{sec:uniqueness}.
Also, the calculations below are arguably less insightful than the general argument -- essentially consisting of a sequence of Gauss sums.

However, given the length of Sec.~\ref{sec:uniqueness}, we felt it was worth to include a more concise proof.

We will frequently write $\omega(x)$ for $\omega_3^x$, for increased readability of complicated exponents.

\subsubsection{Removing the $k$-dependence}

The fact that the $k$-dependence of $\lambda$ as defined in Eq.~(\ref{eqn:sparse}) factors out can be used to remove the summation over $k$ from the auto-correlation conditions.
For the variables $\vec a, \vec b\in\ZZ_6^2$ appearing in the auto-correlations, use the notation
\begin{align*}
	\vec a 
	\simeq
	\begin{pmatrix}
		(k,x)\\
		(l,y)\\
	\end{pmatrix},
	\qquad
	\vec b 
	\simeq
	\begin{pmatrix}
		(r,u)\\
		(s,v)\\
	\end{pmatrix}.
\end{align*}
The auto-correlations become 
\begin{align*}
	\big(\lambda \star \lambda\big)(\vec a)
	&=
	\sum_{r} 
	\omega\big((k+r)^2-r^2\big)
	\sum_{s,u,v}
	\omega\big( \phi(x+u,y+v; l+s)+(l+s)^2-\phi(u,v;s)-s^2 \big)  .
\end{align*}
Shift the summation according to $s \mapsto s-2^{-1} l$.
In $\ZZ_3$ this means that $s\mapsto s+l$ and 
$l+s \mapsto l+s+l=s-l$, so that
\begin{align*}
	\big(\lambda \star \lambda\big)(\vec a)
	&=
	\sum_{r} 
	\omega\big(-rk+k^2\big)
	\sum_{s,u,v}
	\omega\big( 
		\phi(x+u,y+v; s-l)-\phi(u,v;s+l) + (s-l)^2 - (s+l)^2
	\big)  
	\\
	&=
	3
	\delta(k)
	\sum_{s,u,v}
	\omega\big( 
		\Delta
		-ls
	\big),
\end{align*}
having set
\begin{align}\label{eqn:Delta}
	\Delta
	:=
	\phi(x+u, y+v;s-l)-\phi(u,v;s+l).
\end{align}
Analogously, the twisted auto-correlation satisfies
\begin{align*}
	\big(\lambda \twisted \lambda\big) (\vec a)
	&=
	\sum_{r} 
	\omega\big((k+r)^2-r^2 - lr \big)
	\sum_{s,u,v}
	\omega\big(\phi(x+u,y+v;l+s)+(l+s)^2-\phi(u,v;s)-s^2+ks\big) 
	(-1)^{xv-yu} \\
	&= %
	\sum_{r} 
	\omega\big(-r(l +k) + k^2  \big) 
	\sum_{s,u,v}
	\omega\big(\phi(x+u,y+v;l+s)+l^2 + s^2 - ls-\phi(u,v;s)-ls-s^2\big) 
	(-1)^{xv-yu} \\
	&=
	3\delta(k+l)
	\omega^{l^2}
	\sum_{s,u,v}
	\omega\big(\Delta+ls\big) 
	(-1)^{xv-yu} .
\end{align*}
Thus, Eq.~(\ref{eqn:autocorr_conditions}) is equivalent to the reduced conditions
\begin{align}
		\sum_{u,v}
		\sum_s
		\omega\big(\Delta-ls\big)
		&\propto \delta(l)\delta(x)\delta(y), 
		\label{eqn:reduced_proper_cond} 
		\\
		\sum_{u,v}
		(-1)^{xv-yu} 
		\sum_s
		\omega(\Delta+ls)
		&\propto \delta(s)\delta(x)\delta(y)
		\label{eqn:reduced_twisted_cond}
\end{align}
for all $(l,x,y)\in\ZZ_3\times \ZZ_2^2$.

It remains to verify the conditions (\ref{eqn:reduced_proper_cond}, \ref{eqn:reduced_twisted_cond}).
The function $\phi(x,y;l)$ defining Eq.~(\ref{eqn:sparse}) is a explicit polynomial in the variable $l$.
On the other hand, its dependence on $x,y$ is specified distinguishing cases.
Therefore, we will have to treat a number of cases separately.

\subsubsection{The case $(x,y)=(0,0)$}

For the case $(x,y;u,v)=(0,0;1,1)$, the function $\phi$ and hence the difference $\Delta$ vanishes, and we obtain
\begin{align*}
	\sum_{s} \omega(\Delta\mp ls)
	=
	\sum_{s} \omega(\mp ls)
	=
	3\delta(l).
\end{align*}
For the case $(x,y;u,v)=(0,0,u,v)$ with $(u,v)\neq (1,1)$, label $\phi$ by $n=\hat u - \hat v$ to find
\begin{align}\label{eqn:00n}
	\Delta
	=
	\phi(n;s-l)
	-
	\phi(n;s+l)
	=
	(n+s-l)^2
	-
	(n+s+l)^2
	=
	-l(n+s)
\end{align}
so that 
\begin{align*}
	\sum_{s} \omega(\Delta-ls)
	=
	\omega(-nl)
	\sum_{s} \omega(-ls)
	=3\delta(l)
	,
	\qquad
	\sum_{s} \omega(\Delta+ls)
	=
	3
	\omega(-nl)
	.
\end{align*}
In summary:
\begin{center}
	\begin{tabular}{|c|c|c||l|l|}
		\hline
		$x,y$ & $u,v$ & $(-1)^{xv-yu}$ & $\sum_s \omega^{\Delta-ls}$ & $\sum_s\omega^{\Delta+ls}$ \\
		\hline
		00&11  & $1$	&  $3\delta(l)$   & $3\delta(l)$     \\ 
		00&$n$ & $1$	&  $3\delta(l)$   & $3\omega^{-nl}$  \\ 
		\hline
	\end{tabular}
\end{center}
The conditions (\ref{eqn:reduced_proper_cond}, \ref{eqn:reduced_twisted_cond}) for the $(x,y)=(0,0)$-case are immediate.

\subsubsection{The case $(x,y)=(1,1)$}

For $(u,v)=(0,0)$, the first summand in $\Delta$ vanishes, leaving us with
\begin{align*}
	\Delta
	=
	\phi(1,1;s-l) 
	-
	\phi(0,0;s+l)
	= 
	- (s+l)^2
	=
	-s^2+ls - l^2.
\end{align*}
Summing over $s$ leads to quadratic Gauss sums, which can be evaluated using Eq.~(\ref{eqn:gauss_sum_3}) to
\begin{align*}
	\sum_s \omega(\Delta-ls)
	=
	\omega(-l^2)
	\sum_s \omega(-s^2)
	=
	-
	\gamma
	\omega^{-l^2},
	\quad
	\sum_s \omega(\Delta+ls)
	=
	\omega(-l^2)
	\sum_s \omega(-s^2-ls)
	=
	-\gamma,
	\qquad
	\gamma:=\sqrt 3 i.
\end{align*}
Translated to our notation, the trivial relations
\begin{align*}
	\lambda(\vec a + \vec b) \bar \lambda(\vec b)
	=
	\overline{
		\lambda(\vec b)
		\bar\lambda(\vec a + \vec b) 
	},
	\qquad
	\lambda(\vec a + \vec b) \bar \lambda(\vec b)\omega^{[a,b]}
	=
	\overline{
		\lambda(\vec b)
		\bar\lambda(\vec a + \vec b) 
		\omega^{[b,a]}
	}
\end{align*}
say that the substitutions
\begin{align}\label{eqn:ab_symmetry}
	u \leftrightarrow u+x,
	\quad
	v \leftrightarrow v+y,
	\quad
	l \leftrightarrow -l
	\qquad\text{ lead to }\qquad
	(\Delta \mp ls) \leftrightarrow - (\Delta \mp ls),
\end{align}
as can be verified by inspection of Eq.~(\ref{eqn:Delta}).
It follows that the $(u,v)=(1,1)$-term is just the complex conjugate of the $(0,0)$-one.

The cases $(u,v)=(1,0), (0,1)$ can be labeled by $n=\hat u - \hat v=\pm 1$ and give
\begin{align*}
	\Delta
	=
	\phi(-n;s-l) 
	-
	\phi(n;s+l)
	=
	(-n+s-l)^2
	-
	(n+s+l)^2
	=
	-ns
	-ls 
	.
\end{align*}
Performing the character sum, 
\begin{center}
	\begin{tabular}{|c|c|c||l|l|l|}
		\hline
		$x,y$ & $u,v$ & $(-1)^{xv-yu}$ & $\sum_s \omega^{\Delta-ls}$ & $\sum_s\omega^{\Delta+ls}$ \\
		\hline
11&00  & $\phantom{-}1$	& $-\gamma\omega^{-l^2}$   & $-\gamma$   \\ 
11&$n$ & $-1$           & $\phantom{-}3\delta(n-l)$  	  & $\phantom{-}0$ \\ %
11&11  & $\phantom{-}1$	& $\phantom{-}\gamma\omega^{l^2}$   & $\phantom{-}\gamma$   \\ 
		\hline
	\end{tabular}
\end{center}
Summing over $u,v$, we get the conditions 
(\ref{eqn:reduced_proper_cond}, \ref{eqn:reduced_twisted_cond}),
with the only non-trivial case being the proper auto-correlation for $l=\pm 1$, which produces the sum
\begin{align*}
	\sqrt{3} (i \omega(1) - i \omega(-1))+3.
\end{align*}
Using the explicit value $\sin(2\pi/3)=\sqrt 3/2$, one finds that this sum vanishes, as required.

\subsubsection{The cases $(x,y)=(1,0), (0,1)$}

Let $n=\hat x - \hat y\in \{\pm 1\}$.
For $(u,v)=(0,0)$, get
\begin{align*}
	&&
	&\Delta
	=
	\phi(n;s-l)
	-
	\phi(0,0;s+l)
	=
	(n+s-l)^2 - (s+l)^2
	=
	1+nl %
	-s(l+n),
	\\
	&\Rightarrow&
	&\sum_s \omega(\Delta-ls) 
	=	
	3
	\omega^2 %
	\delta(n-l),
	\qquad
	\sum_s \omega(\Delta+ls) 
	=
	0.
\end{align*}
The cases $(x,y;u,v)=(n;n)$,
i.e.\ those where $(x+u, y+v)=(0,0)$, give again the complex conjugate, by Eq.~(\ref{eqn:ab_symmetry}).

For $(x,y;u,v)=(n;-n)$ apply Eq.~(\ref{eqn:gauss_sum_3}) to arrive at
\begin{align*} 
	&&
	&\Delta
	=
	\phi(11;s-l)
	-
	\phi(-n;s+l)
	=
	- (-n+s+l)^2
	=
	-s^2
	+
	(l-n)s
	-l^2
	-ln
	-1
	,
	\\
	&\Rightarrow&
	&\sum_s \omega(\Delta-ls) 
	=
	-\gamma \omega\big(
		n^2 %
		-l^2
		-ln
		-1
	\big)
	=
	-\gamma \omega\big(
		-l^2
		-ln
	\big),
	\\
	&&
	&\sum_s \omega(\Delta+ls) 
	=
	-\gamma \omega\big(
		(l+n)^2 %
		-l^2
		-ln
		-1
	\big)
	=
	-\gamma \omega( ln).
\end{align*}

A final application of (\ref{eqn:ab_symmetry}) to get the $(x,y;u,v)=(-n;1,1)$-case 
leads to the following table, which is treated like the previous one:
\begin{center}
	\begin{tabular}{|c|c|c||l|l|l|}
		\hline
		$x,y$ & $u,v$ & $(-1)^{xv-yu}$ &  $\sum_s \omega^{\Delta-ls}$ & $\sum_s\omega^{\Delta+ls}$ \\
		\hline
$n$&00   & $\phantom{-}1$	& $\phantom{-}3\omega^2\delta(n-l)$    & $\phantom{-}0$ \\
$n$&$n$  & $\phantom{-}1$	& $\phantom{-}3\omega^{-2}\delta(n-l)$ & $\phantom{-}0$ \\
$n$&$-n$ & $-1$           & $-\gamma \omega^{-l^2-ln}$           & $-\gamma\omega(ln)$ \\
$-n$&$11$& $-1$           & $\phantom{-}\gamma \omega^{l^2+ln}$  & $\phantom{-}\gamma\omega(-ln)$ \\
		\hline
	\end{tabular}
\end{center}

\section{Symmetries of the auto-correlation equations}
\label{sec:symmetries}

\subsection{List of symmetries}
\label{sec:list_of_symmetries}

The auto-correlation conditions, Eq.~(\ref{eqn:autocorr_conditions}), have a number of symmetries, i.e.\ operations that map solutions to solutions.
Some are listed below. 
If applicable, we give implementations in terms of manifestly two-unitarity preserving operations acting on  $U_\lambda$.

\begin{enumerate}

	\item
		{}[\emph{Symplectic maps on phase space}].
		For 
		$S\in \Sp(V)$,
		there is a symmetry 
		$\lambda(\vec a) \mapsto \lambda(S^{-1}\vec a)$.
		It can be implemented by choosing an element $U$ of the Clifford group associated with $G$ and conjugating $U_\lambda$ by $U\otimes \bar U$:
		\begin{align*}
			(U \otimes \bar U) \, U_\lambda \, (U^\dagger \otimes U^t)
			&=
			\sum_{\vec a}
			\lambda(\vec a) 
			(U \otimes \bar U) \, |\Phi_{\vec a}\rangle\langle\Phi_{\vec a}| \, (U^\dagger \otimes U^t)
			\\
			&=
			\sum_{\vec a}
			\lambda(\vec a) 
			(U w(\vec a) U^\dagger \otimes \Id) \, |\Phi\rangle\langle\Phi| \, (U w(\vec a)^\dagger U^\dagger \otimes \Id)  %
		\\
		&=
			\sum_{\vec a}
			\lambda(\vec a) 
			(w(S\vec a) \otimes \Id) \, |\Phi\rangle\langle\Phi| \, ( w(S\vec a)^\dagger \otimes \Id)
		=
			\sum_{\vec a}
			\lambda(S^{-1} \vec a) 
			\, |\Phi_{\vec a}\rangle\langle\Phi_{\vec a}|.
		\end{align*}

	\item
		{}[\emph{Extended symplectic maps on phase space}].\label{itm:extended}
		Let $PT: (\vec p, \vec q) \mapsto (\vec p, -\vec q)$.
		Then
		$\lambda(\vec a) \mapsto \lambda(PT^{-1}\vec a)$ is a symmetry, which can be implemented by conjugating with the flip operator:
		\begin{align}\label{eqn:pt}
			\begin{split}
			\FF \, U_\lambda \, \FF
			&=
			\sum_{\vec a}
			\lambda(\vec a) 
			\FF
			\, |\Phi_{\vec a}\rangle\langle\Phi_{\vec a}| \, 
			\FF
			=
			\sum_{\vec a}
			\lambda(\vec a) 
			(\Id\otimes w(\vec a) ) \, |\Phi\rangle\langle\Phi| \, (\Id \otimes w(\vec a))^\dagger  %
		\\
		&=
			\sum_{\vec a}
			\lambda(\vec a) 
			(w(\vec a)^t \otimes \Id) \, |\Phi\rangle\langle\Phi| \, ( w(\vec a)^t \otimes \Id)^\dagger 
		=
			\sum_{\vec a}
			\lambda(PT^{-1} \vec a) 
			\, |\Phi_{\vec a}\rangle\langle\Phi_{\vec a}|.
		\end{split}
		\end{align}
		Together with $\Sp(V)$, the $PT$-symmetry generates an action of $\ESp(V)$ on phase space points (because for every $S\in \ESp(V)$, either $S$ or $PT\,S$ is an element of $\Sp(V)$).
		In particular ``time reversal'' $T: (\vec p, \vec q)\mapsto (-\vec p, \vec q)$ is a symmetry.
		For $n=1$, this means that any linear map with determinant $\pm1$ is implementable,
		and for $n=1$ and $d=3$, any element of $\GL(\ZZ_3^2)$.

	\item
		{}[\emph{Linear shifts on phase space}].
		For any $\vec b\in V$, there is a symmetry
		$\lambda(\vec a) \mapsto \lambda(\vec a - \vec b)$.
		It can be implemented by
		\begin{align*}
			(w(\vec b) \otimes \Id) \, U_\lambda (w(\vec b)^\dagger \otimes \Id)
			&=
			\sum_{\vec a}
			\lambda(\vec a) 
			(w(\vec b) w(\vec a) \otimes \Id) 
			\,|\Phi\rangle\langle\Phi|\,
			(w(\vec a)^\dagger w(\vec b)^\dagger \otimes \Id)  \\
			&=
			\sum_{\vec a}
			\lambda(\vec a) 
			\omega^{[\vec b,\vec a]}
			(w(\vec b+\vec a) \otimes \Id) 
			\,|\Phi\rangle\langle\Phi|\,
			(w(\vec b+\vec a)^\dagger \otimes \Id)  
			\omega^{[\vec a,\vec b]}
			\\
			&=
			\sum_{\vec a}
			\lambda(\vec a-\vec b) \, |\Phi_{\vec a}\rangle\langle\Phi_{\vec a}|
			= 
			(\Id\otimes w(\vec b)^t) \, U_\lambda (\Id\otimes \bar w(\vec b))
			.
		\end{align*}
		The first three symmetries together generate the action of the \emph{affine extended symplectic group} 
		$V \rtimes \ESp(V)$ on phase space points.

	\item
		{}[\emph{Multiplication by a character}].
		For any $\vec b\in V$, there is a symmetry
		$\lambda(\vec a) \mapsto \omega^{[\vec b, \vec a]} \lambda(\vec a)$.
		It can be implemented by
		\begin{align*}
			(w(\vec b) \otimes \bar w(\vec b)) \, U_\lambda 
			&=
			\sum_{\vec a}
			\lambda(\vec a) 
			(w(\vec b) w(\vec a) w(\vec b)^\dagger \otimes \Id) \, |\Phi\rangle\langle\Phi_{\vec a}|
		=
			\sum_{\vec a}
			\omega^{[\vec b,\vec a]}
			\lambda(\vec a) \, |\Phi_{\vec a}\rangle\langle\Phi_{\vec a}|
			= 
			 U_\lambda 
			( w(\vec b) \otimes  \bar w(\vec b)) \,
			.
		\end{align*}

	\item\label{itm:field_autos}
		{}[\emph{Field automorphisms}].
		Complex conjugation
		$\lambda \mapsto \bar \lambda$
		is a symmetry.
		To implement it, first conjugate the unitary to get
		\begin{align*}
			\overline{U_\lambda}
			=
			\sum_{\vec a}
			\bar\lambda(\vec a) 
			(w(T \vec a) \otimes \Id) 
			\, |\Phi\rangle\langle\Phi| \,
			(w(T \vec a) \otimes \Id)^\dagger
			=
			\sum_{\vec a}
			\bar\lambda(T^{-1}\vec a) 
			\, |\Phi_{\vec a}\rangle\langle\Phi_{\vec a}|,
		\end{align*}
		and then undo the action of $T\in\ESp(V)$.
		More generally, assume that $\lambda$ takes values in the cyclotomic field $\QQ[\tau_d]$.
		Then Galois automorphisms act on doubly perfect functions by
		$\lambda(\vec a) \mapsto (k \lambda)(S^E_k \vec a)$.
		The action is implemented by semi-linear Galois Clifford maps in the same ways as described above for $\Sp(V)$.
		Note:
		Unlike $\ESp(V)$, it does not follow that the action of $\GSp(V)$ permuting phase space points alone is a symmetry.

	\item
		{}[\emph{Global phases}].
		Trivially, $\lambda \mapsto c \lambda$ for for $c\in\CC, |c|=1$ is a symmetry.

	\item
		{}[\emph{Fourier transform}].
		The $\ZZ_d^{2n}$-Fourier transform
		\begin{align*}
			\lambda \mapsto \mathcal{F} \lambda,
			\qquad
			\big(\mathcal{F}\lambda)(\vec a)
				&=
				\frac1{d^n} 
				\sum_{\vec b}
				\omega_d^{-[\vec a, \vec b]} \lambda(\vec b)
		\end{align*}
		is a symmetry.
		We give a short proof in Lem.~\ref{lem:ft_sym} below, and an alternative argument in Sec.~\ref{sec:characteristic}.
\end{enumerate}

\subsubsection{The Fourier transform as a symmetry}

It remains to show:

\begin{lemma}\label{lem:ft_sym}
	If $\lambda: V\to\CC$ is doubly perfect, then so is $\F \lambda$.
\end{lemma}

\begin{proof}
	Recall the \emph{convolution theorem}, which says
	\begin{align*}
		\overline{(\F f)}
		\,
		(\F g)
		&=
		\frac1{d^n}
		\,
		\F (f \star g). 
	\end{align*}
	It implies that the Fourier transform satisfies the unitarity and dual-unitarity conditions:
	\begin{align*}
		\big|\F \lambda\big|^2
		=
		\F(\lambda \star \lambda)
		=
		\frac1d
		\F(d^{2n}\, \delta)
		=
		1,
		\qquad
		(\F \lambda) \star (\F \lambda)
		=
		d^n
		\,
		\F^{-1}
		\big(
			\overline{(\F f)}
			\F f
		\big)
		=
		d^n
		\,
		\F^{-1} (1)
		= 
		d^{2n}\,\delta.
	\end{align*}
	For the twisted auto-correlation,  compute
	\begin{align}\label{eqn:twisted_symmetry}
		\begin{split}
			\big((\F \lambda) \twisted (\F \lambda)\big)(\vec a)
			&=
			\sum_{\vec b} 
			(\F \lambda)(\vec a + \vec b)
			\overline{(\F \lambda)(\vec b)} 
			\,
			\omega^{[\vec a,\vec b]}
			\\
			&=
			\frac1{d^{2n}}
			\sum_{\vec b} 
			\sum_{\vec k, \vec l}
			\lambda(\vec k)
			\bar\lambda(\vec l)
			\,
			\omega^{-[\vec a+\vec b,\vec k] + [\vec b, \vec l] +[\vec a, \vec b]}
			\\
			&=
			\sum_{\vec k, \vec l}
			\lambda(\vec k)
			\bar\lambda(\vec l)
			\,
			\omega^{-[\vec a,\vec k]}
			\,
			\frac1{d^{2n}}
			\sum_{\vec b} 
			\omega^{[\vec k - \vec l + \vec a, \vec b]}
			\\ %
			&=
			\sum_{\vec l}
			\lambda(\vec l - \vec a)
			\bar\lambda(\vec l)
			\,
			\omega^{-[\vec a,\vec l]}
			=
			(\lambda \twisted \lambda)(-\vec a) = d^{2n}\,\delta(\vec a).
		\end{split}
	\end{align}
\end{proof}

Remarks:
\begin{itemize}
	\item
		All known doubly perfect functions take values that are roots of unity.
		This property is preserved under the FT.
		More precisely, 
		because the elements of the Schur matrix lie in $\QQ[\tau_d]$, 
		if $d'$ is a multiple of $d$ such that the values of $\lambda$ lie in $\QQ[\tau_{d'}]$, 
		then the same is true for $\F \lambda$.
		In particular, the Fourier transforms of the artisanal solutions, and the solutions reported in \cite{Rather2024Construction}, are all given in terms of powers of $\omega_6$.
\end{itemize}

\subsubsection{Computer-found solution under symmetries}

It turns out that the computer solution $\Lambda_3$ of Ref.~\cite{Rather2024Construction} can be mapped to the symmetric artisanal one by this sequence: %
\begin{enumerate}
	\item
		Multiply by the character associated with the vector $\vec b=(2,2)$,
	\item
		Apply the linear map 
		$
			G=
			\begin{pmatrix}
				3 & 5 \\
				1 & 2
			\end{pmatrix}
			\in 
			\SL(\ZZ_6^2)
		$,
	\item
		Shift by $\vec b=(3,3)$.
\end{enumerate}
This could indicate that there are only few orbits -- possibly only one -- of doubly perfect functions of order six that take values in powers of $\omega_3$.

\subsection{Auto-correlation conditions and the characteristic function}
\label{sec:characteristic}

The WH operators form an orthogonal basis with respect to the Hilbert-Schmidt inner product:
\begin{align*}
	\tr w(\vec a)^\dagger w(\vec b) = d^n\,\delta_{\vec a, \vec b}.
\end{align*}
The map 
\begin{align*}
	c_B(\vec a) := d^{-n/2} \, \tr w(\vec a)^\dagger B
\end{align*}
that sends a phase space point $\vec a\in V$ to the expansion coefficient of an operator $B$ in that basis is the \emph{characteristic function} \cite{gross2006hudson}.

For a rank-one operator $B=|\psi\rangle\langle\phi|$,
\begin{align*}
	c_{|\psi\rangle\langle\phi|}(\vec p, \vec q)
	=
	\langle\phi|w(-\vec p, -\vec q) |\psi\rangle
	=
	d^{-n/2}
	\tau^{-\vec p \vec q}
	\sum_{\vec b}
	\omega^{-\vec p(\vec b-\vec q) }
	\bar\phi_{\vec b - \vec q}
	\psi_{\vec b}
	=
	d^{-n/2}
	\tau^{\vec p \vec q}
	\sum_{\vec b}
	\omega^{-\vec p \vec b }
	\bar\phi_{\vec b - \vec q}
	\psi_{\vec b},
\end{align*}
which looks similar to the cross-correlation functions considered above.

Indeed, for the special case 
\begin{align*}
	B=|\bar\lambda\rangle\langle\bar\lambda|, \qquad
	|\bar\lambda\rangle = \sum_{\vec a} \bar\lambda(\vec a)|\vec a\rangle \in (\CC^d)^{\otimes 2n},
	\qquad
	\vec p = -J \vec q,
\end{align*}
we recover the twisted correlation function
\begin{align*}
	d^{-n/2}\,
	c_{|\bar\lambda\rangle\langle\bar\lambda|}(-J \vec q, \vec q)
	&=
	\tau^{-\vec q J \vec q}
	\sum_{\vec b}
		\lambda_{\vec b - \vec q}
		\bar\lambda_{\vec b}
		\,
		\omega^{(J\vec q) \vec b}
	=
	\sum_{\vec b}
		\lambda_{\vec b - \vec q}
		\bar\lambda_{\vec b}
		\,
		\omega^{-[\vec q, \vec b]}
	=
	\sum_{\vec b}
		\lambda_{\vec b + \vec q}
		\bar\lambda_{\vec b}
		\,
		\omega^{[\vec q, \vec b]}
	=
	(\lambda\twisted\lambda)(\vec q).
\end{align*}
Likewise,
\begin{align*}
	d^{-n/2}\,
	c_{|\bar\lambda\rangle\langle\bar\lambda|}(0, \vec q)
	&=
	\sum_{\vec b}
		\lambda_{\vec b - \vec q}
		\bar\lambda_{\vec b}
	=
	(\lambda\star\lambda)(\vec q),
\end{align*}
so that both auto-correlation functions considered in this paper appear as restrictions of the characteristic function.

This gives us another way of thinking about the identity (\ref{eqn:twisted_symmetry}):
The Fourier transform is a Clifford operation,
the characteristic function is covariant under the Clifford group \cite{gross2006hudson, gross2021schur},
and the symplectic map $J$, which is associated with the Fourier transform,
acts linearly on the set
\begin{align}\label{eqn:twisted_char}
	\{ (-J \vec q, \vec q)^t \,|\, \vec q \in \ZZ_d^{2n}\}
	=
	\range
	\begin{pmatrix}
		-J \\
		\Id
	\end{pmatrix}
	=
	\ker 
	\begin{pmatrix}
		\Id & J
	\end{pmatrix}
	\subset \ZZ_d^{4n}
\end{align}
of points which describes the twisted auto-correlation function.

More precisely, the vector $|\bar\lambda\rangle$ is an element of $(\CC^d)^{\otimes {2n}}$,
and a Clifford unitary acting on this space is therefore associated with an element $S$ of $\Sp(\ZZ_d^{4n})$.
Writing $S$ as a block matrix, the conditions for it to leave the set (\ref{eqn:twisted_char}) invariant read
\begin{align*}
	0
	=
	\begin{pmatrix}
		\Id & J
	\end{pmatrix}
	\begin{pmatrix}
		A & B \\
		C & D
	\end{pmatrix}
	\begin{pmatrix}
		-J\\
		\Id
	\end{pmatrix}
	=
	-(A + JC)J+ B + JD
	=
	-AJ - JCJ+ B + JD.
\end{align*}
This is solved by the symplectic representation $A=D=0, B=-C=\Id$ of the Fourier transform,
and by the action of $G\in\Sp(\ZZ_d^{4n})$, embedded into $\Sp(\ZZ_d^{4n})$ as $A=D^{-t}=G, B=C=0$.

\section{Hadamard Two-Unitaries}
\label{sec:hadamard}

\subsection{Hadamard two-unitaries from doubly perfect functions}

Doubly perfect functions give rise to Hadamard two-unitaries:

\hadamard*

\begin{proof}
	We first consider $H$.
	The matrix $C$ with entries $C_{\vec x, \vec y}=d^{-n} \, \lambda(\vec x - \vec y)$ is circulant (with respect to addition in $\ZZ_d^{2n}$). 
	Circulant matrices are diagonal in the Fourier basis, with eigenvalues the Fourier-transform of the first row, i.e.\
	\begin{align*}
		C = (F\otimes F) D (F \otimes F)^\dagger,
		\qquad
		D = \sum_{\vec a\in V} (\F \lambda)(\vec a)\,|\vec a\rangle\langle \vec a|,
	\end{align*}
	where $F$ is the Fourier gate defined in Eq.~(\ref{eqn:schur}).
	By Eq.~(\ref{eqn:phase_gates}), the phase factors in the definition of $H$ are controlled-$Z$ gates between the $i$th and the $(n+i)$th system for $i=1, \dots, n$.
	Denoting the product of these gates by $CZ$,
	\begin{align*}
		H = 
		CZ
			(F\otimes F)\, D \, (F\otimes F)^\dagger 
			CZ^\dagger.
	\end{align*}
	The following gate identity is well-known, and can be checked using the symplectic matrices in Sec.~\ref{sec:concrete_cliffords}:
	\begin{align*}
		CZ (F\otimes F) 
		=
		(F^\dagger\otimes\Id)
		CX
		=
		(F^\dagger\otimes\Id) U_{WH},
	\end{align*}
	where $CX$ stands for the product of controlled-$X$ gates, with the $(n+i)$th system controlling the $i$th one.
	Hence
	\begin{align}\label{eqn:diag_h}
		H = 
		(F\otimes\Id)^\dagger
		U_{\F \lambda} 
		(F\otimes\Id).
	\end{align}
	The gate $U_{\F\lambda}$ is two-unitary because the FT of a doubly perfect function is doubly perfect by Sec.~\ref{sec:symmetries}.
	But then $H$, being locally equivalent to a two-unitary, is two-unitary as well.

	If $\lambda$ is perfect, then so is $\lambda \circ PT$, where 
	$PT(\vec p, \vec q) = (\vec p, -\vec q)$ (c.f.\ Sec.~\ref{sec:symmetries}).
	Let $H$ be the two-unitary constructed from $\lambda \circ PT$ constructed above.
	Then its partial transpose is also two-unitary.
	It has matrix elements
	\begin{align*}
		(H^\Gamma)_{\vec a_1,\vec a_2;\vec b_1,\vec b_2}
		=
		H_{\vec a_1,\vec b_2;\vec b_1,\vec a_2}
		=
		\omega^{\vec a_1 \vec b_2 - \vec a_2 \vec b_1} \lambda\big(PT(\vec a_1-\vec b_1, \vec b_2-\vec a_2) \big)
		=
		\omega^{[\vec a, \vec b]} \lambda(\vec a - \vec b)=G_{\vec a, \vec b}.
	\end{align*}
\end{proof}

Remark:
\begin{itemize}
	\item
	From Eq.~(\ref{eqn:diag_h}) and the definition of $U_\lambda$, we can read off the eigendecomposition of $H$:
	\begin{align}\label{eqn:h_eigen}
		H =
		\sum_{\vec a\in V}
		(\F \lambda)(\vec a)
		|\tilde\Phi_{\vec a}\rangle \langle\tilde\Phi_{\vec a}|,
		\qquad
		|\tilde\Phi_{(\vec p, \vec q)}\rangle 
		= (F^\dagger\otimes\Id) |\Phi_{(\vec p, \vec q)}\rangle
			.
	\end{align}
	These maximally entangled states are known as the two-qudit \emph{cluster states} \cite{briegel2001persistent, zhou2003quantum}.
\end{itemize}

\subsection{Construction of doubly perfect functions}

Theorem~\ref{thm:hadamard} motivates the search for doubly perfect functions of order other than $6$.

As remarked after Eq.~(\ref{eqn:perfect_quadratic}), the quadratic exponential
\begin{align*}
	\lambda(\vec a) = \tau^{\vec a N \vec a}, \qquad \vec a \in \ZZ_d^{2n}, N=N^t \in \ZZ^{2n\times 2n} 
\end{align*}
is doubly perfect if and only if both $N$ and $N+J$ have trivial kernel over $\ZZ_d$.
If $d$ is odd, then $N=\Id$ provides a simple solution for arbitrary $n$.
Thus the complex Hadamard matrix
\begin{align*}
	G_{\vec a, \vec b}
	=
	\tau_d^{(\vec a - \vec b) (\vec a - \vec b)} \omega_d^{[\vec a, \vec b]}
\end{align*}
is proportional to a two-unitary of order $d^{n}$.

The case $d=2^n$, $n>1$ is more interesting.

First, heuristically, we do not expect that ``generating a two-unitary''  to be a rare property among quadratic forms.
Indeed, the probability that a uniformly random, symmetric binary $n\times n$ matrix is non-singular is well-known \cite{brent2010determinant}.
As $n\to\infty$, it converges from above to a known limit value $p_\infty>.419$.
If we heuristically assume that the event that $N$ is non-singular is 
approximately independent
of the event that $N+J$ is non-singular,
then a fraction of $p_\infty^2 > 1/6$ all binary symmetric matrices $N$ should generate two-unitaries.

A rigorous construction that works for all $n>1$ is as follows.

\begin{theorem}\label{thm:even_doubly_perfect}
	There are doubly perfect functions on the phase space $V=\ZZ_2^{2n}$ for every $n>1$.
\end{theorem}

\begin{proof}
	Consider $\FF_{2^n}$ as an $n$-dimensional vector space over $\FF_2$.
	By \cite[Theorem~4]{seroussi1980factorization},
	there exists an ortho-normal basis 
	$\{b_1, \dots, b_n\}\subset \FF_{2^n}$ 
	for the \emph{trace form}
	$(a, b)_{\tr} = \tr_{\FF_{2^n}/\FF_2}(ab)$ .
	Let $0,1\neq \alpha\in\FF_{2^n}$. %
	Then $G_{ij}=(b_i, \alpha b_j)_{\tr}$ is manifestly symmetric.
	But because the basis is self-dual with respect to the trace form, $G$ is also a matrix representation of the action by multiplication of $\alpha$ on $\FF_{2^n}$.
	Therefore, 
	\begin{align*}
		N = 
		\begin{pmatrix}
			G & 0 \\
			0 & G
		\end{pmatrix},
		\quad
		N+J=
		\begin{pmatrix}
			G & \Id \\
			\Id & G
		\end{pmatrix}
		\qquad
		\text{ are matrix representations of }
		\qquad
		\begin{pmatrix}
			\alpha & 0 \\
			0 & \alpha
		\end{pmatrix},
		\quad
		\begin{pmatrix}
			\alpha & 1 \\
			1 & \alpha
		\end{pmatrix}
	\end{align*}
	respectively.
	But the determinants of the $2\times 2$-matrices are $\alpha^2$ and $\alpha^2-1$, both of which are non-zero by assumption on $\alpha$.
	The claim now follows by the criterion discussed after Eq.~(\ref{eqn:perfect_quadratic}).
\end{proof}

Remarks:
\begin{itemize}
	\item
More generally, if $A, B\in\ZZ_2^{n\times n}$  then \cite{silvester2000determinants}
\begin{align*}
	\det
	\begin{pmatrix}
		A & \Id \\
		\Id & B
	\end{pmatrix}
	=
	\det(AB-\Id)
	.
\end{align*}
Therefore, if $A, B$ are a pair of non-singular symmetric matrices such that $AB-\Id$ is also non-singular, one obtains a suitable quadratic form by taking $N$ to be the block-diagonal matrix with $A, B$ on the main diagonal.
As a concrete example, computer experiments indicate that the matrices with ``kite-shaped support'', $A_{ij}=B_{ij}=1$ iff $i+j\leq n+1$, give a solution for all $n$ not congruent to $1$ modulo $3$. %
	\item
		If $n$ is not a multiple of four, there even is a \emph{normal} trace-orthonormal basis \cite[Result~1.6]{jungnickel1993trace}.
		Presumably, this choice will result in two-unitaries with particularly regular structure.
	\item
		If $d$ is of the form $d=p_0^n$ for an odd prime $p_0$ and an odd exponent $n$, 
		then Theorem~4 of Ref.~\cite{seroussi1980factorization} again shows the existence of a trace-orthogonal basis of $\FF_{p_0^n}$ over $\FF_{p_0}$, so that the proof of 
		Thm.~\ref{thm:even_doubly_perfect} extends to this case.
		This gives an alternative construction to the one presented for any odd order at the beginning of this section.
	\item
		Using the fact that tensor products 
		(in the sense of Sec.~\ref{sec:products}) of doubly perfect functions are again doubly perfect,
		we can combine the constructions of this section and the artisanal solutions.
		This gives analytic examples of complex Hadamard two-unitaries for every order $d$, 
		unless $d$ is of the form $d=2d_1$, where $d_1$ is neither divisible by $2$ nor by $3$.
		It seems plausible that doubly perfect functions also exist for these orders,
		and we leave their construction as an open problem.
\end{itemize}

\subsection{Generalized circulant matrices}
\label{sec:circulant}

The matrices $H, G$ have additional regularity properties beyond being complex Hadamard matrices and proportional to two-unitaries.

A matrix $C$ is \emph{circulant} if it commutes with cyclic shifts $Cw(0,q)=w(0,q)C$.
Let's call it circulant with respect to addition in $\ZZ_d^n$, if $Cw(0,\vec q)=w(0,\vec q)C$
for all $\vec q\in\ZZ_d^n$.

Now consider a matrix that is diagonal in some stabilizer basis.
Then it obviously commutes with the stabilizer group of the basis.
Because the stabilizer group is isomorphic to  $\{w(0,\vec q)\}_{\vec q}$, that property generalizes the notion of cyclicity.
One could refer to such matrices as being \emph{phase space} or \emph{time-frequency} or \emph{twisted} \emph{circulant}.
As pointed out above, both $U_\lambda$ and $H$ have this property.
The stabilizer group of $U_\lambda$ has been given in Eq.~(\ref{eqn:wh_stabilizer}):
\begin{align}\label{eqn:ulambda_stab}
	\{ w(\vec p_1,-\vec p_1,-\vec p_2,-\vec p_2) \},
	\quad
	\vec p_1, \vec p_2 \in \ZZ_2^n,
	\qquad\text{ generated by }\qquad
	X_i\otimes X_{i+n},
	\quad
	Z_i\otimes Z_{i+n}^\dagger
	.
\end{align}
The stabilizer group of $H$ follows by applying $J^{-1}$ on the first subsystem, resulting in
\begin{align*}
  \{ w(\vec p_1, \vec p_2,-\vec p_2, -\vec p_1) \},
	\quad 
	\vec p_1, \vec p_2\in\ZZ_2^n,
	\qquad\text{ generated by }\qquad
	Z_i \otimes X^\dagger_{n+i},
	\quad
	X^\dagger_{i}\otimes Z_{n+i}.
\end{align*}

The matrix $G$ also shows a certain circularity property.
Recall that a matrix $C$ is circulant with respect to addition in $\ZZ_d^n$ if its $\vec q$-th column arises from the $0$-th column by an application of $w(0,\vec q)$:
\begin{align*}
	C |\vec q\rangle = w(0,\vec q) (C|0\rangle).
\end{align*}
For the matrix $G$, too, the $\vec q$-th column can be generated from the $0$-th one, but this time by a 
phase space shift
\begin{align*}%
	G|\vec q\rangle
	&=
	\sum_{\vec a-\vec q}
	\lambda(\vec a -\vec q) \omega^{[\vec a, \vec q]} |\vec a\rangle
	=
	\sum_{\vec a-\vec q}
	\lambda(\vec a) \omega^{\vec q_2 \vec a_1 - \vec q_1\vec a_2 } |\vec a\rangle
	=
	\sum_{\vec a}
	\lambda(\vec a-\vec q) 
	Z^{\vec q_2}\otimes Z^{-\vec q_1}  |\vec a\rangle \\
	&=
	\sum_{\vec a}
	\lambda(\vec a) 
	(Z^{\vec q_2}\otimes Z^{-\vec q_1})
	\,
	(X^{\vec q_1}\otimes X^{\vec q_2})
	|\vec a\rangle 
	=
	(Z^{\vec q_2}\otimes Z^{-\vec q_1})
	\,
	(X^{\vec q_1}\otimes X^{\vec q_2})
	G|0\rangle 
	.
\end{align*}

\section{Algebraic formulation}
\label{sec:algebra}

In Sec.~\ref{sec:definition}, we have introduced the theory notion of a two-unitary in terms of operations on matrix indices.
In this section, we sketch an alternative approach with a more algebraic flavor.

\subsection{Quasi-orthogonal subalgebras}

Recall that a matrix algebra forms a Hilbert space in its own right, with inner product given by the Hilbert-Schmidt form.
Specifically, we will use the normalization
\begin{align}\label{eqn:tracial}
	(X|Y) := 
	\tau
	\big(
		X^\dagger Y
	\big),
	\qquad
	\text{for matrices}
	\quad
	X, Y \in \M_d,
	\quad
	\text{in terms of the \emph{tracial state}}
	\quad
	\tau(A)
	=
	\frac1{d}\,
	\tr A.
\end{align}
Given a matrix algebra $\A$, let $\mathcal{A}_0=\{ X \in \A \,|\, \tr X = 0\}$ be the subspace of trace-free elements.
Two subalgebras $\mathcal{A}, \mathcal{B}$ are \emph{quasi-orthogonal} 
\cite{Ohno2007QuasiOrthogonal, ohno2008quasi, weiner2010quasi}
(or \emph{complementary} \cite{Petz2007Complementary, petz2010algebraic})
if $\A_0, \B_0$ are 
orthogonal. 

\begin{example}
	Two examples for the case 
	$\M_d\otimes \M_d$:
	\begin{enumerate}
		\item
		The local observable algebras
		$\mathcal{L}$ and $\R$ are quasi-orthogonal.
		\item
		If $\H=\CC^2$, the (commutative) algebra $\A=\{\Id\otimes \Id, Z\otimes Z\}$,
		generated by the tensor product of Pauli-$z$-matrices,
		is quasi-orthogonal to both $\mathcal{L}$ and $\R$.
		Quasi-orthogonality thus captures the fact that the correlations measured by $Z\otimes Z$ are not locally accessible.
	\end{enumerate}
\end{example}

The characterization of two-unitarity in terms of quasi-orthogonality then reads:
\begin{proposition}\label{prop:quasi_orthogonal}
	Let 
	$U\in \M_d\otimes \M_d$
	be a unitary.
	Then we have the equivalences
	\begin{align*}
		&(1) &
		U^\Gamma\text{ is unitary}
			 &\quad\Leftrightarrow\quad
		U\mathcal{L}U^\dagger\text{ is quasi-orthogonal to }\mathcal{R} &
																																		&\quad\Leftrightarrow\quad
		U\mathcal{R}U^\dagger\text{ is quasi-orthogonal to }\mathcal{L}, \\
			 &(2) &
		U^R\text{ is unitary}
			 &\quad\Leftrightarrow\quad 
		U\mathcal{L}U^\dagger\text{ is quasi-orthogonal to }\mathcal{L} &
			 &\quad\Leftrightarrow\quad
		U\mathcal{R}U^\dagger\text{ is quasi-orthogonal to }\mathcal{R},\\
			 &(3) &
		U\text{ is two-unitary}
			 &\quad\Leftrightarrow\quad
		U\mathcal{L}U^\dagger\text{ is quasi-orthogonal to }\mathcal{L}\text{ and }\mathcal{R} &
																																													 &\quad\Leftrightarrow\quad
		U\mathcal{R}U^\dagger\text{ is quasi-orthogonal to }\mathcal{L}\text{ and }\mathcal{R}.
 \end{align*}
\end{proposition}

This leads to a purely algebraic version of the classification problem of two-unitaries.

\algebraic*

In principle, the equivalence allows one to demonstrate the existence of a two-unitary just by exhibiting a suitable subalgebra.

\begin{example}
	In $\M_d\otimes \M_d$, the matrices $Z\otimes Z$ and $X\otimes X$
	generate a subalgebra $\A$ with linear basis
	\begin{align*}
		(Z^p X^q)\otimes (Z^p X^q),
		\qquad
		p,q\in\ZZ_d.
	\end{align*}
	With the exception of $(p,q)=(0,0)$ (which corresponds to the operator $\Id\otimes\Id$), the basis elements are clearly traceless and orthogonal to both $\mathcal{L}$ and $\R$.
	The generators satisfy the same algebraic relations as $Z^2$ and $X^2$.
	Because for odd $d$ the operators $Z^2, X^2$ generate $\M_d$, it follows that $\A$ is isomorphic to $\M_d$.

	The construction thus witnesses the existence of two-unitaries of any odd order.
	(In fact, $\A$ is the image of $\mathcal{L}$ under the two-unitary that corresponds to the OLS construction given in Eq.~(\ref{eqn:linear_ols}) for $\alpha=-1$). %
\end{example}

Remark:
\begin{itemize}
	\item
		The two-unitarity condition can also be expressed as
		\begin{align}\label{eqn:two_unitarity_direct_sum}
			U(\mathcal{L}_0 \oplus \R_0)U_\lambda^\dagger \,\perp\, (\mathcal{L} \oplus \R),
			\quad\text{or}\quad
			U(\mathcal{L} \oplus \R)U_\lambda^\dagger \,\perp\, (\mathcal{L}_0 \oplus \R_0),
			\quad\text{or}\quad
			U(\mathcal{L}_0 \oplus \R_0)U_\lambda^\dagger \,\perp\, (\mathcal{L}_0 \oplus \R_0)
			.
		\end{align}
\end{itemize}

\subsection{Proof of the algebraic characterization}
\label{sec:algebraic_proof}

The proof will use a quantitative measure of the similarity of two subalgebras.
It was previously studied independently in Ref.~\cite{gross2012index} (as $\eta$, see below) and in Ref.~\cite{weiner2010quasi} (as $c=\eta^2$).

For a matrix algebra $\A$, let $P_{\A}$ be the Hilbert-Schmidt projection onto $\A$.
Then the \emph{overlap} between two matrix subalgebras $\A, \B$ is
\begin{align*}
	\eta(\A,\B) := \sqrt{\Tr P_\A P_\B}.
\end{align*}
(Note that $P_\A, P_\B$ are ``superoperators'' in quantum jargon, and consequently, the trace is the one over Hilbert-Schmidt space).

Then two unital subalgebras $\A, \B\subset B(\K)$ are quasi-orthogonal if and only if $\eta(\A, \B)=1$.
More generally, we have the following result of Ref.~\cite[Section 7]{gross2012index}.
It makes use of the normalized \emph{Schatten $4$-norm} $\|A\|_{4,\tau}$ of an operator, defined by
\begin{align}\label{eqn:schatten_4}
	\|A\|_{4,\tau}^4 = 
	\tau\big(( A A^\dagger )^2\big).
\end{align}

\begin{lemma}[\cite{gross2012index}]\label{lem:overlaps}
	Let $\H_L, \H_R$ be finite-dimensional Hilbert spaces of dimensions $d_L, d_R$.
	Let $\mathcal{L}=B(\H_L)\otimes \Id \subset B(\H_L\otimes \H_R)$
	be the observable algebra associated with the first subsystem, and define $\R$ analogously.
	\begin{itemize}
		\item
			For a unitary $U: \H_L\otimes \H_R \to \H_L\otimes \H_R$,
			\begin{align*}
				\eta(U \mathcal{L} U^\dagger, \R)^2
				=
				\|U^\Gamma\|_{4,\tau}^4.
			\end{align*}
		\item
			For a unitary $U: \H_L\otimes \H_R \to \H_R\otimes \H_L$,
			\begin{align*}
				\eta(U \mathcal{L} U^\dagger, \mathcal{L})^2
				=
				\|U^R\|_{4,\tau}^4.
			\end{align*}
	\end{itemize}
\end{lemma}

The second claim, which was not directly stated in Ref.~\cite{gross2012index}, follows from the substitution $U\mapsto U \FF$ and Eq.~(\ref{eqn:realignment_flip}).

\begin{proof}[Proof \emph{(of Prop.~\ref{prop:quasi_orthogonal})}]
	The normalized Schatten 4-norm of an operator $A$ on a $d^2$-dimensional space is equal to $d^{-1/2}$ times the $\ell_4$-norm of the vector $\vec \sigma$ of its singular values.
	Likewise, $\|A\|_{2,\tau}=\sqrt{\tau(A A^\dagger)}$
	is $d^{-1}$ times the $\ell_2$-norm of $\sigma$.
	In dimension $d^2$, we have the standard norm inequality $d^{1/2} \|\sigma\|_{\ell_4} \geq \| \sigma \|_{\ell_2}$ with equality if and only if all singular values are equal.
	Therefore,
	unitaries  are characterized by the condition 
	\begin{align*}
		A\text{ unitary}
		\qquad\Leftrightarrow\qquad
		\|A\|_{4,\tau}
		=
		\|A\|_{2,\tau}
		= 1.
	\end{align*}
	Because the partial transpose permutes matrix elements, it preserves the 2-norm, so that for every unitary $U$, we automatically have $\|U^\Gamma\|_{2,\tau} = 1$.
	Thus, $U^\Gamma$ is unitary if and only if $\|U^\Gamma\|_{4,\tau}=1$.
	By Lem.~\ref{lem:overlaps}, this is equivalent to $\eta(U \mathcal{L} U^\dagger, \R)^2 = 1$, i.e.\ that $U\mathcal{L} U^\dagger$ is quasi-orthogonal to $\mathcal{L}$.

	The other cases are treated analogously.
\end{proof}

\begin{proof}[Proof \emph{(of Thm.~\ref{thm:algebraic})}]
	If $U\in \M_d\otimes \M_d$ and $U'=U(V_L\otimes V_R)$ for unitaries $V_L, V_R\in \M_d$, then 
	\begin{align*}
		U'
		\mathcal{L} 
		(U')^\dagger
		=
		U 
		(V_L \M_d V_L^\dagger\otimes V_R V_R^\dagger)
		U^\dagger
		=
		U (\M_d \otimes \Id) U^\dagger
		=
		U \mathcal{L} U^\dagger
		=:\A
	\end{align*}
	is a unital subalgebra isomorphic to $\M_d$.
	If $U$ is a two-unitary, then $\A$ is quasi-orthogonal to $\mathcal{L}, \R$ by Lem.~\ref{lem:overlaps}.

	Conversely, let $\A\subset \M_d\otimes \M_d$ be a unital subalgebra isomorphic to $\M_d$ that is quasi-orthogonal to $\mathcal{L}, \R$.
	Then the commutant $\A'$ is also isomorphic to $\M_d$ \cite[Lem.~11.8]{takesaki2003theory}.
	Choose $*$-isomorphisms
	$\alpha_L: \mathcal{L}\to\A$,
	$\alpha_R: \mathcal{R}\to\A'$.
	Then $\alpha_L\otimes\alpha_R$ defines a representation of $\M_d\otimes \M_d$ with multiplicity $1$.
	By 
	\cite[Thm.~11.9]{takesaki2003theory},
	there exists a unitary $U$ such that
	\begin{align*}
		\alpha_L(A) \otimes \alpha_R(B)
		=
		U (A \otimes B) U^\dagger
		.
	\end{align*}
	Then $U$ is two-unitary by Prop.~\ref{prop:quasi_orthogonal}.
\end{proof}

\subsection{Proof of the generalized auto-correlation conditions}
\label{sec:auto_proof}

Here, we present a short proof of Ref.~\cite{Rather2024Construction}'s auto-correlation conditions (\ref{eqn:autocorr_conditions}) in the algebraic picture, and generalize them to the case where $n\geq 1$.

As per the remark in Sec.~\ref{sec:wh_basis}, 
instead of working with 
$\H\otimes\H=(\CC^d)^{\otimes n}\otimes(\CC^d)^{\otimes n}$
and the WH basis $\{|\Phi_{\vec a}\rangle\}_{\vec a \in V}$,
we realize the bi-partite Hilbert space directly as 
$\H\otimes\H=\M_{d^n}$,
equipped with the normalized Hilbert-Schmidt inner product and the WH operators $\{w(\vec a)\}_{\vec a\in V}$ as ortho-normal basis.
(This approach is not new, see e.g.\ Ref.~\cite{tyson2003operator}).

In this picture, the algebra $\mathcal{L}$ consists of $\M_{d^n}$, (seen as operators) acting by left-multiplication on $\M_{d^n}$, (seen as vectors).
The algebra $\mathcal{R}$ is $\M_{d^n}$, acting by right-multiplication of the adjoint.

Specifically, choose elements $\vec a, \vec b, \vec c\in V$.
Then the action $L_{\vec a}$ of $w({\vec a})\in\mathcal{L}$ 
and the action $R_{\vec c}$ of $w({\vec c})\in\mathcal{R}$ on a basis element $w(\vec b)\in\H$ are, respectively
\begin{align}
	\begin{split}\label{eqn:wh_action}
		L_{\vec a} w(\vec b) &= w({\vec a}) w(\vec b)\phantom{w(\vec c)^\dagger} = \tau^{[{\vec a},\vec b]} w({\vec a}+\vec b), \\
		R_{\vec c} w(\vec b) &= \phantom{w(\vec a)}w(\vec b) w({\vec c})^\dagger = \tau^{-[\vec b,{\vec c}]} w(\vec b-{\vec c}).
	\end{split}
\end{align}

\begin{example}
	As a consistency check, let's verify that the expressions (\ref{eqn:wh_action}) for the left- and right action on WH operators commute:
	\begin{align*}
		L_{\vec a} R_\vec c w(\vec b)
		&=
		L_{\vec a}
		\,
		\tau^{-[\vec b,\vec c]} w(\vec b-\vec c) 
		=
		\tau^{-[\vec b,\vec c]+[\vec a,\vec b]-[\vec a, \vec c]} 
		w(\vec a + \vec b-\vec c)
		\\
		R_{\vec c} L_{\vec a} w(\vec b)
		&=
		R_{\vec c}
		\,
		\tau^{[{\vec a},\vec b]} w({\vec a}+\vec b) 
		=
		\tau^{[{\vec a},\vec b] -[{\vec a},\vec c]-[\vec b,\vec c] } 
		w({\vec a}+\vec b-\vec c)
	\end{align*}
	which are indeed equal.
	Next, check that the algebras are quasi-orthogonal:
	\begin{align*}
		(L_{-\vec a}|R_\vec c)
		&=
		\frac1{d^{2n}}
		\tr L_{\vec a} R_\vec c 
		=
		\frac1{d^{2n}}
		\sum_\vec b \tau\big( w(\vec b)^\dagger L_{\vec a} R_\vec c w(\vec b)  \big)
		=
		\frac1{d^{2n}}
		\delta_{{\vec a},\vec c}
		\,
		\sum_\vec b \tau^{[{\vec a},\vec b]-[\vec a+\vec b,{\vec a}]}
		=
		\frac1{d^{2n}}
		\delta_{{\vec a},\vec c}
		\,
		\sum_\vec b \omega^{[{\vec a},\vec b]}
		= 
		\delta(\vec a) \delta(\vec c).
	\end{align*}
\end{example}

Now define the WH-diagonal unitary
\begin{align*}
	U_\lambda: w(\vec b) \mapsto \lambda({\vec b}) w(\vec b).
\end{align*}
Conjugating $L_{\vec a}\in \mathcal{L}$ by $U_{\lambda}$ gives an operator that acts on a basis element as
\begin{align*}
	U_\lambda L_{\vec a} U_\lambda^\dagger\, 
	w(\vec b)
	=
	\lambda({{\vec a}+\vec b})\bar\lambda({\vec b})
	\,
	\tau^{[{\vec a},\vec b]}
	\,
	w({\vec a}+\vec b)
\end{align*}
so that we obtain the inner products
\begin{align*}
	\big(L_{\vec a}\big| U_\lambda L_{\vec a'} U_\lambda^\dagger\big)
	&=
	\frac1{d^{2n}}
	\delta_{\vec a,\vec a'}
	\sum_{\vec b}
\lambda(\vec a+\vec b) \bar\lambda(\vec b)
	\,
	\tau^{[\vec a,\vec b]} 
	\tau^{[-\vec a,\vec a+\vec b]}
	=
	\frac1{d^{2n}}
	\delta_{\vec a,\vec a'}
	\sum_{\vec b}
\lambda(\vec a+\vec b) \bar\lambda(\vec b),
	\\
	\big(R_{-\vec c}\big| U_\lambda L_{\vec a} U_\lambda^\dagger\big)
	&=
	\frac1{d^{2n}}
	\delta_{\vec a,\vec c}
	\sum_{\vec b}
\lambda(\vec a+\vec b) \bar\lambda(\vec b)
	\,
	\tau^{-[\vec a+\vec b,\vec a]}
	\tau^{[\vec a,\vec b]} 
	=
	\frac1{d^{2n}}
	\delta_{\vec a,\vec c}
	\sum_{\vec b}
\lambda(\vec a+\vec b) \bar\lambda(\vec b)\omega^{[\vec a, \vec b]}.
\end{align*}
Quasi-orthogonality is then equivalent to the inner product evaluating to $\delta(\vec a)$, which is Eq.~(\ref{eqn:autocorr_conditions}), as required.

\subsection{Limits on the entangling power of unitaries diagonal in a stabilizer basis}

Recall the discussion from Sec.~\ref{sec:great_wall}:
Because there is no two-unitary of order $2$, 
it follows that any two-unitary of even order $d=2d_1$ must entangle the two-dimensional and the $d_1$-dimensional subsystems.

However, conversely, we will now show that unitaries $U_\lambda$ of the form given in Eq.~(\ref{eqn:perfect_wh})  are limited in the extent to which they can entangle internal degrees of freedom.
On a high level, the reason is that the symmetry group exhibited in Sec.~\ref{sec:circulant} factorizes, thereby giving rise to local conserved quantities.
In this sense, it might be surprising that two-unitaries of order 6 of that restricted form can be found.

We will measure the ``degree of entanglement created'' in terms of \emph{support algebras} 
(also: \emph{interaction algebras})
\cite{zanardi2000stabilizing, gross2012index}.
Let $X$ be an element of a tensor product $\A\otimes \B$ of algebras.
The \emph{support algebra}
$\Spp(X, \A)$ 
of $X$ in $\A$ 
is the smallest 
$*$-subalgebra of $\A$ such that $X \in \Spp(X, \A) \otimes \mathcal{B}$.
For a subset $\Omega\subset \A\otimes \B$, define 
$\Spp(\Omega,\A)$
as the $*$-algebra generated by 
$\bigcup_{X\in\Omega} \Spp(X,\A)$,
the supports
of all its elements.

Now assume that $d=d_1 d_2$ is the product of two co-prime numbers
(for simplicity, we restrict to the case $n=1$).
The Chinese Remainder unitary $R$ of Sec.~\ref{sec:great_wall} establishes an isomorphism
\begin{align*}
	\M_{d}\otimes \M_d \simeq 
	\big(\M_{d_1}\boxtimes \M_{d_1}\big)
	\otimes
	\big(\M_{d_2}\boxtimes \M_{d_2}\big),
\end{align*}
where we have again used boxed tensor product symbols to visually indicate internal degrees of freedom.
By Eq.~(\ref{eqn:chinese_wh}), the 
symmetry group of $U_\lambda$,
given in Eq.~(\ref{eqn:ulambda_stab}), likewise factorizes 
\begin{align*}
	G\simeq G_{1} \boxtimes G_{2},
	\quad
	G_{i} 
	= 
	\langle X\boxtimes X, Z^{\kappa_i} \boxtimes Z^{-\kappa_i} \rangle 
	= 
	\langle X\boxtimes X, Z \boxtimes Z^{-1} \rangle 
	\subset \M_{d_i}\boxtimes \M_{d_i}.
\end{align*}
With these notions, we have the following lemma (in the spirit of \cite[Lemma~7]{gross2012index}):

\begin{lemma}\label{lem:support}
	Assume that $d=d_1d_2$ is the product of two co-prime numbers.
	For any phase function $\lambda:\ZZ_d^2\to \CC$,
	it holds that the support algebra of $U_\lambda (\M_{d_1}\boxtimes \M_{d_1}) U_\lambda^\dagger$
	within $\M_{d_2} \boxtimes \M_{d_2}$
	is diagonal in the WH basis:
	\begin{align*}
		\Spp\big(
			U_\lambda (\M_{d_1}\boxtimes \M_{d_1}) U_\lambda^\dagger,
			\M_{d_2} \boxtimes \M_{d_2}
		\big)
		\subset 
		\Span G_{2}
		=
		\Big\{ 
			\sum_{\vec a\in\ZZ_{d_2}^2} c_{\vec a} \, |\Phi_{\vec a}\rangle\langle \Phi_{\vec a}| \,\big|\, c_{\vec a} \in \CC
		\Big\}
		.
	\end{align*}
\end{lemma}

\begin{proof}
	Denoting the commutant of an algebra $\A$ by $\A'$ and using that, by Sec.~\ref{sec:circulant}, $U_\lambda$ commutes with $G_2$,
	\begin{align*}
		\big(U_\lambda (\M_{d_1}\boxtimes \M_{d_1}) U_\lambda^\dagger \big)'
		=
		U_\lambda (\M_{d_2}\boxtimes \M_{d_2}) U_\lambda^\dagger 
		\supset
		\Id
		\otimes
		U_\lambda G_2 U_\lambda^\dagger
		=
		\Id
		\otimes
		G_2.
	\end{align*}
	Because $G_2$ is a matrix group, the $*$-algebra it generates is equal to its linear span.
	Because it is a maximal stabilizer group, the algebra is given by the set of operators diagonal in the associated stabilizer basis, i.e.\ the WH basis.
	It is also a maximal abelian subalgebra (MASA) of $\M_{d_2}\boxtimes \M_{d_2}$, and as such equal to its own commutant.
	Hence, invoking the double commutant theorem,
	\begin{align*}
		U_\lambda (\M_{d_1}\boxtimes \M_{d_1}) U_\lambda^\dagger 
		\subset
		\big(
			\Id
			\otimes
			G_2
		\big)'
		=
		(\M_{d_1}\boxtimes \M_{d_1})
		\otimes
		\Span G_2.
	\end{align*}
\end{proof}

Colloquially speaking, under unitaries diagonal in a stabilizer basis, the qutrit and the qubit subsystems can exchange at most one maximally abelian subalgebra.

\subsection{Algebraic description of the artisanal solution}

In this section, we will compute the images 
$U_\lambda \mathcal{L} U_\lambda^\dagger$, 
$U_\lambda \mathcal{R} U_\lambda^\dagger$
of the local observable algebras for the artisanal function $\lambda$. 

\subsubsection{Qutrit algebras}
\label{sec:qutrit_algebras}

We start with the images of the subalgebra of local observables 
$\mathcal{L}^{(3)}$,
$\mathcal{R}^{(3)}$.
These can be read off directly from Eq.~(\ref{eqn:23cliff}):
\begin{align*}
	\begin{split}
		U_\lambda \mathcal{L}^{(3)} U_\lambda^\dagger
		\quad\text{ has linear basis }\quad
		&w(p+q, p+q)\otimes w(-q, p) \otimes V_t w(p,0) V_t^\dagger
		+
		w(p-q, p+q)\otimes w(q, p) \otimes |\Phi_{11}\rangle\langle\Phi_{11}|, \\
		U_\lambda \mathcal{R}^{(3)} U_\lambda^\dagger
		\quad\text{ has linear basis }\quad
		&w(-q, p)\otimes w(p+q, p+q)\otimes  V_t w(p,0) V_t^\dagger
		+
		w(q, p)\otimes w(p-q,p+q) \otimes |\Phi_{11}\rangle\langle\Phi_{11}|.
	\end{split}
\end{align*}

Remarks:
\begin{itemize}
	\item
		Plainly,
		swapping the left and the right algebras also swaps their images.
		This can be traced back to the fact that $\lambda$ is independent of the sign of $k$, c.f.\ Eq.~(\ref{eqn:pt}).
	\item
		The support algebra  
		$
		\Spp\big(
			U_\lambda (\M_3\boxtimes \M_3) U_\lambda^\dagger,
			\M_2 \boxtimes \M_2
		\big)
		$
		of the qutrit system inside the qubit system is $\CC^4$, realized as the algebra of matrices diagonal in the qubit WH basis.
		This exhausts the constraints imposed by Lem.~\ref{lem:support}.
\end{itemize}

\subsubsection{The $\mathfrak{so}(4;\CC)$-picture}
\label{sec:lie}

The images of the qubit algebras
$\mathcal{L}^{(2)}$,
$\mathcal{R}^{(2)}$
have a more complicated structure than the qutrit ones.

We will find it advantageous to start describing the space $U_\lambda(\mathcal{L}_0^{(2)}\oplus \R_0^{(2)}) U_\lambda^\dagger$.
It turns out that this direct sum contains a basis that is easier to work with than is the case for any of the two summands individually.

None of the spaces
$\mathcal{L}_0, \R_0$ 
or
$\mathcal{L}_0\oplus \R_0$ 
is an associative algebra.
But they all are complex Lie algebras, isomorphic to $\mathfrak{sl}(d; \CC)$ and 
$\mathfrak{sl}(d; \CC)\oplus\mathfrak{sl}(d; \CC)$ respectively.
This point of view is particularly fruitful for $d=2$, where we can make use of the well-known isomorphism
\begin{align}\label{eqn:qubit_lie}
	\mathfrak{sl}(2; \CC)\oplus\mathfrak{sl}(2; \CC)
	\simeq
	\mathfrak{so}(4; \CC).
\end{align}
(The algebra
$\mathfrak{so}(4;\CC)$ may be better-known to physicists as the complexification $\mathfrak{so}(1,3)_\CC$ of the Lie algebra of the Lorentz group).
The isomorphism is realized particularly cleanly by choosing a basis for $\CC^2\otimes\CC^2$ that consists of a vectorization of the usual representation of the quaternions:
\begin{align}
	\begin{split}\label{eqn:quaternio_vec}
		(-iw(1,1)\otimes\Id)|\Phi\rangle 
		&= -\phantom{i}|\Phi_{1,1}\rangle =: |\Phi_{\emptyset}\rangle,  \\
		(-iw(0,1)\otimes\Id)|\Phi\rangle
		&=-i|\Phi_{0,1}\rangle =: |\Phi_{-1}\rangle, \\
		(w(0,0)\otimes\Id)|\Phi\rangle
		&=\phantom{-i}|\Phi_{0,0}\rangle =: |\Phi_{0}\rangle, \\
		(-iw(1,0)\otimes\Id)|\Phi\rangle
		&=-i|\Phi_{1,0}\rangle =: |\Phi_{1}\rangle.
	\end{split}
\end{align}
The labels on the right hand side have been chosen so that $|\Phi_{\emptyset}\rangle$ spans the singlet space, and the subscript in $|\Phi_{-1}\rangle, |\Phi_{0}\rangle, |\Phi_{1}\rangle$ corresponds to the variable $m$ in the triplet space, in the sense that
\begin{align*}
	V^\dagger |\Phi_m\rangle = |m\rangle \in \CC^3.
\end{align*}
It is then easy to verify that the action of the trace-less local observable algebras is represented in this basis exactly by the complex anti-symmetric $4\times 4$-matrices.
We will work with the following commonly used basis for the Lie algebra $\mathfrak{so}(4)$,
\begin{align}
	\begin{split}\label{eqn:so4basis}
	J_{-1}
	&:=
	\left(
		\begin{array}{cccc}
			0 & 0 & 0 & 0 \\
			0 & 0 & 0 & 0 \\
			0 & 0 & 0 & -1 \\
			0 & 0 & 1 & 0 \\
	 \end{array}
	\right),
	\quad
	J_0
	:=
	\left(
		\begin{array}{cccc}
			0 & 0  & 0 & 0  \\
			0 & 0  & 0 & 1 \\
			0 & 0  & 0 & 0  \\
			0 & -1 & 0 & 0 
	 \end{array}
	\right),
	\quad
	J_1
	:=
	\left(
		\begin{array}{cccc}
			0  & 0 & 0  & 0 \\
			0  & 0 & -1 & 0 \\
			0  & 1 & 0  & 0 \\
			0  & 0 & 0  & 0 \\
	 \end{array}
	\right), \\
	K_{-1}
	&:=
	\left(
		\begin{array}{cccc}
			0  & 1 & 0 & 0 \\
			-1 & 0 & 0 & 0 \\
			0  & 0 & 0 & 0 \\
			0  & 0 & 0 & 0 \\
	 \end{array}
	\right),
	\quad
	K_0
	:=
	\left(
		\begin{array}{cccc}
			0  & 0  & 1 & 0  \\
			0  & 0  & 0 & 0 \\
			-1 & 0  & 0 & 0  \\
			0  & 0  & 0 & 0 
	 \end{array}
	\right),
	\quad
	K_1
	:=
	\left(
		\begin{array}{cccc}
			0  & 0 & 0  & 1 \\
			0  & 0 & 0  & 0 \\
			0  & 0 & 0  & 0 \\
			-1 & 0 & 0  & 0 \\
	 \end{array}
	\right).
	\end{split}
\end{align}
The Lie brackets between the basis elements are
\begin{align*}
	[J_{-1}, J_0] = J_1,
	\quad
	[K_{-1}, K_0] = K_1,
	\quad
	[K_{-1}, K_0]= J_1
	\qquad
	\text{and cyclic permutations thereof}.
\end{align*}
It is straight-forward to check that the left / right trace-less observables map to the commuting Lie subalgebras with basis
\begin{align}\label{eqn:local_basis_jk}
	R_m = J_m + K_m,
	\qquad
	L_m = J_m - K_m.
\end{align}

\subsubsection{Image of the qubit algebras}
\label{sec:qubit_algebras}

Using the notions of Sec.~\ref{sec:unitaries}, define
\begin{align*}
	D_2 := 
	U^Q_{N_2},
	\qquad
	D_{3,m} :=
	(\Id\otimes\Id\otimes\langle m|)
		U^Q_{N_3}
	(\Id\otimes\Id\otimes|m\rangle),
	\qquad
	W_m
	:= 
	U_{WH} \,(D_{3,m} D_2^\dagger)\, U_{WH}^\dagger.
\end{align*}
Then we see that $U_\lambda$ amounts to a ``controlled-$W_m$ gate'' in the sense that
\begin{align}\label{eqn:controlled_wm}
	U_\lambda\, (\Id\otimes |\Phi_m\rangle\langle\Phi_\emptyset|) \,U_\lambda^\dagger
	=
	W_m \otimes |\Phi_m\rangle\langle\Phi_\emptyset| .
\end{align}
Taking the adjoint and products, this implies
\begin{align*}
	\begin{array}{rcrcl}
	U_\lambda\, (\Id\otimes |\Phi_\emptyset\rangle\langle\Phi_n|) \,U_\lambda^\dagger
	&=&
		W_n^\dagger &\otimes& |\Phi_\emptyset\rangle\langle\Phi_n| ,
	\\
			U_\lambda\, (\Id\otimes |\Phi_m\rangle\langle\Phi_n|) \,U_\lambda^\dagger
								&=&
		W_m W_n^\dagger &\otimes& |\Phi_m\rangle\langle\Phi_n|, \\
			U_\lambda\, (\Id\otimes |\Phi_\emptyset\rangle\langle\Phi_\emptyset|) \,U_\lambda^\dagger
										&=&
		\Id&\otimes& |\Phi_\emptyset\rangle\langle\Phi_\emptyset|
			.
	\end{array}
\end{align*}
A basis for image of the right / left qubit algebra respectively is
(using $\mathrm{h.c.}$ to denote the Hermitian conjugate) is therefore
\begin{align*}
	\begin{array}{lcrclcllll}
		J_1 \pm K_1
		&=&
		W_{1} W_{0}^\dagger &\otimes& |\Phi_{1}\rangle\langle \Phi_{0}|
												&\pm&
		W_{-1}^\dagger &\otimes& |\Phi_\emptyset\rangle\langle\Phi_{-1}| 
												&-\quad\mathrm{h.c.}
		\\
		J_2 \pm K_2
												&=&
		-
		W_{1} W_{-1}^\dagger &\otimes& |\Phi_{1}\rangle\langle \Phi_{-1}|
												 &\pm&
		W_{0}^\dagger &\otimes& |\Phi_\emptyset\rangle\langle\Phi_{0}| 
												 &-\quad\mathrm{h.c.}
		\\
		J_3 \pm K_3
												 &=&
		W_{0} W_{-1}^\dagger &\otimes& |\Phi_{0}\rangle\langle \Phi_{-1}|
												 &\pm&
		W_{1}^\dagger &\otimes& |\Phi_\emptyset\rangle\langle\Phi_{1}| 
												 &-\quad\mathrm{h.c.}
	\end{array}
\end{align*}

All information about the image of the qubit observable algebra under $U_\lambda$ is thus contained in the three operators $W_m \in U(\CC^3\otimes\CC^3)$.
For the sparse artisanal solution,
\begin{align}
		&&
		D_{3,m} D_2^\dagger|k,l\rangle
		&=
		\omega^{Q(k,l,m)}
		|k,l\rangle
		= 
		\omega^{l^2-lm+m^2} 
		|k,l\rangle
		=
		(\Id\otimes (U^Q_{(2)} Z^{-m}))
		|k,l\rangle
		,
		\\
		\label{eqn:w_preimage}
		&\Rightarrow&
		W_m
		&= 
		\omega^{m^2} 
		U_{WH}
		(\Id\otimes (U^Q_{(2)} Z^{-m}))
		U_{WH}^\dagger 
		.
\end{align}
\subsubsection{Quasi-orthogonality for the qubit algebra}

In preparation of the proof in Sec.~\ref{sec:uniqueness}, we give the conditions for
$U_\lambda(\mathcal{L}_0^{(2)} \oplus \mathcal{R}_0^{(2)})U_\lambda^\dagger$  
to be orthogonal to $\mathcal{L}\oplus\R$ in the $\mathfrak{so}(4)$-picture.
All calculations in this section will be re-done in Sec.~\ref{sec:uniqueness} in greater generality.

The projection of 
\begin{align*}
	U_\lambda 
	K_m
	U_\lambda^\dagger
	&=
	W_m^\dagger \otimes |\Phi_\emptyset\rangle\langle\Phi_m| 
	-
	W_m \otimes |\Phi_m\rangle\langle\Phi_\emptyset|
\end{align*}
onto $\Id\otimes \mathfrak{so}(4)$ vanishes if and only if 
\begin{align}\label{eqn:k_condition}
	W_m^\dagger + W_m
	\quad\text{is orthogonal to}\quad 
	\mathcal{L}^{(3)}\otimes\mathcal{R}^{(3)}.
\end{align}
Likewise, the condition for the $J$-type elements of the basis reads, for $m\neq n$,
\begin{align}\label{eqn:j_condition}
	W_mW_n^\dagger + W_nW_m^\dagger
	\quad\text{is orthogonal to}\quad 
	\mathcal{L}^{(3)}\otimes\mathcal{R}^{(3)}.
\end{align}

We now check these conditions for the sparse artisanal solution.
By Eq.~(\ref{eqn:w_preimage}), the operator
$D_{3,m} D_2^\dagger$ 
acts non-trivially only on the second qutrit, where it is diagonal.
It can therefore be expanded in terms of $Z$-type WH operators:
\begin{align}\label{eqn:wm_exp}
	D_{3,m} D_2^\dagger
	=
	(U^Q_{(2)} Z^{-m}))
	&=
	\Id\otimes
	\sum_{r\in\ZZ_3} f_m(r) \, Z^r, %
\end{align}
where the expansion coefficients follow from a Fourier transform that can be evaluated using Eq.~(\ref{eqn:gauss_sum_3}):
\begin{align}\label{eqn:wm_ft}
	&&
	f_m(r)
	&=
	\frac13
	\sum_{l} \omega(l^2 -lm +m^2- rl)
	=
	\frac{i}{\sqrt 3} \omega(-r^2+rm)
	\quad\Rightarrow\quad
	D_{3,m} D_2^\dagger
	=
	\frac{i}{\sqrt 3} 
	\Id\otimes
	\left(
		\Id
		+
		\bar
		\omega
		\big(
			\omega^m Z + \bar\omega^{m} Z^\dagger
		\big)
	\right).
\end{align}
By
Eq.~(\ref{eqn:swh}),
the image of $\Id \otimes Z = w(0,1,0,0)$ 
under conjugation by $U_{WH}$
is 
$X^\dagger\otimes X^\dagger$, so that
\begin{align}\label{eqn:wm}
	W_{m}
	&=
	\frac{i}{\sqrt 3} 
		\Id
		+
	\frac{i}{\sqrt 3} 
	\bar\omega
		\big(
			\omega^m (X\otimes X)^\dagger + \bar\omega^m X\otimes X
		\big)
		.
\end{align}
The $X\otimes X$-type terms are manifestly orthogonal to the local observables,
and the $i\Id$-type terms cancel when the adjoint is added, so Eq.~(\ref{eqn:k_condition}) is satisfied.
Next,
\begin{align}\label{eqn:artisanal_k_type}
	W_m W_n^\dagger
	=
	\omega^{m^2-n^2}
	U_{WH}
	\,
	(
		\Id
		\otimes
		Z^{-m+n}
	)
	\, U_{WH}^\dagger 
	=
	\omega^{m^2-n^2}
	(X\otimes X)^{n-m},
\end{align}
which manifestly satisfies Eq.~(\ref{eqn:j_condition}).

Remarks:
\begin{itemize}
	\item
		The support algebra  
		$
		\Spp\big(
			U_\lambda (\M_2\boxtimes \M_2) U_\lambda^\dagger,
			\M_3 \boxtimes \M_3
		\big)
		$
		of the qubit system inside the qutrit system is $\CC^3$, realized as the algebra of matrices diagonal in the $X\otimes X$-basis.
		This is strictly smaller than what the constraints imposed by Lem.~\ref{lem:support} allow for.
		It would be interesting to decide whether Lem.~\ref{lem:support} can be strengthened.
\end{itemize}

\section{Proof of the main theorem}
\label{sec:uniqueness}

We will classify those doubly perfect functions $\lambda$ that can be constructed by choosing one quadratic form $P$ on $\ZZ_3^2$ and one quadratic form $Q$ on $\ZZ_3^3$ in the framework given in Sec.~\ref{sec:artisanal}.
Some of the arguments make weaker assumptions.

\subsection{Classification of the choices on the singlet space}
\label{sec:u2}

The artisanal construction started off with an order-$3$ dual-unitary derived from the doubly perfect function $\lambda(k,l)=k^2+l^2$.
This is no loss of generality.
In the next lemma, we show that any two-unitary that can be written as the direct sum of two Cliffords with respect to a decomposition of $\CC^2\otimes \CC^2 = \CC^3\oplus \CC$ must be such that two-qutrit Clifford is two-unitary (of order 3).
Lemma~\ref{lem:only_one_perfection} then establishes that within the doubly perfect function ansatz, there is only one orbit of Clifford two-unitaries of order $3$.

\begin{lemma}\label{lem:u2_is_two_unnitary}
	Let $V: \CC^3\to \CC^2\otimes \CC^2$ be an isometry
	and
	let $P$ be the projection onto the ortho-complement of $\range V$.
	Assume that $U_2, U_3$ are Clifford unitaries acting on two and on three qutrits respectively and set
	\begin{align*}
		U
		= 
		U_2 \otimes P
		+
		(\Id^{\otimes 2} \otimes V)
		U_3 
		(\Id^{\otimes 2} \otimes V)^\dagger.
	\end{align*}
	Then 
	\begin{align*}
		U(\mathcal{L}^{(3)}_0\oplus\R^{(3)}_0)U^\dagger
		\quad\text{is orthogonal to}\quad\mathcal{L}^{(3)}\oplus \R^{(3)}
				 \Leftrightarrow
				 \left\{
					 \begin{array}{ll}
						U_2(\mathcal{L}^{(3)}_0\oplus\R^{(3)}_0)U_2^\dagger
						\quad\text{is orthogonal to}\quad& \mathcal{L}^{(3)}\oplus \R^{(3)}, \\
						U_3(\mathcal{L}^{(3)}_0\oplus\R^{(3)}_0)U_3^\dagger
						\quad\text{is orthogonal to}\quad& \mathcal{L}^{(3)}\oplus \R^{(3)}.
			\end{array}
		\right.
	\end{align*}
	In particular, if $U$ is two-unitary (of order $6$), then $U_2$ is two-unitary (of order $3$).
\end{lemma}

\begin{proof}
	There are symplectic matrices $S_2, S_3$ and 
	vectors $\vec b_1, \vec b_2$ such that,
	for every
	$\vec a\in\ZZ_3^4$,  
	\begin{align*}
		U (w(\vec a) \otimes \Id) U^\dagger
		=
		U_2 (w(\vec a) \otimes P)  U_2^\dagger
		+
		V U_3  (w(\vec a) \otimes V^\dagger V) U^\dagger_3 V^\dagger
		=
		\omega^{[\vec b_2, \vec a]}
		w(S_2 \vec a)\otimes P
		+
		\omega^{[\vec b_3, \vec a]}
		V w(S_3 (\vec a\oplus 0)) V^\dagger.
	\end{align*} 
	Now choose another vector $\vec c\in\ZZ_3^4$ and
	project onto $w(\vec c)\otimes \Id$ to get
	\begin{align*}
		\omega^{[\vec b_2, \vec a]}
		\tr w(\vec c)^\dagger w(S_2 \vec a)
		+
		\omega^{[\vec b_3, \vec a]}
		\tr
		(w(\vec b\oplus 0))^\dagger
		w(S_3 (\vec a\oplus 0)).
	\end{align*}
	The first summand has absolute value $0$ or $3^2$, while the second summand has absolute value $0$ or $3^3$.
	Thus, the expression vanishes if and only if both summands do.
	The claim follows by choosing $\vec a, \vec c$ so that $w(\vec a)$ and $w(\vec c)$ lie in 
	$\mathcal{L}^{(3)}\oplus \R^{(3)}$.
\end{proof}

Next we show that all qutrit Clifford two-unitaries of the form Eq.~(\ref{eqn:perfect_wh}) are equivalent to the $U_2$ chosen in the artisanal solution.

\begin{lemma}\label{lem:only_one_perfection}
	Let $\lambda:\ZZ_3^2\to\CC$ be a doubly perfect function such that $U_\lambda$ is Clifford.
	Then 
	$\lambda$ 
	can be mapped to 
	$\lambda'((k,l))=k^2+l^2$, by a change of global phase,
	multiplication by a character, 
	and the application of an element of $\GL(\ZZ_3^2)$.
\end{lemma}

\begin{proof}
	By Eq.~(\ref{eqn:diagonal_clifford}), we may assume that $\lambda(\vec a)$ is of the form $\omega_3^{2^{-1}\vec a N \vec a + [\vec b, \vec a] + c}$.
	Adjusting a global phase and multiplying with a character if necessary, we may assume that $\vec b=c = 0$.
	From the discussion in Sec.~\ref{sec:ansatz}, we know that both $\det N$ and $\det N + 1$ must be non-zero modulo 3, which implies $\det N=1$.
	Quadratic forms over finite fields are characterized up to basis change by their rank and discriminant.
	Over $\ZZ_3$, the discriminant of a full-rank form is just its determinant.
	But $2\Id$ also has determinant equal to $1$ modulo $3$, so there exits a matrix in 
	$\GL(\ZZ_3^2)$
	such that $G N G^t=2\Id$.
\end{proof}

Remark:
\begin{itemize}
	\item
		Lemma~\ref{lem:only_one_perfection} is less general than the related statement of Ref.~\cite[Theorem~1]{Rather2023Absolutely}, which applies to all order-$3$ two-unitaries.
		However, on this restricted set, it is more powerful, because it realizes the equivalence 
		by a $\GL(\ZZ_3^2)$-action, rather than the application of general local unitaries.
\end{itemize}

\subsection{Classification of the choices on the triplet space}

Having chosen $P(k,l)=k^2+l^2$ without loss of generality,
we now classify the quadratic forms $Q$ on $\ZZ_3^3$ that lead to a doubly perfect function.

Throughout, we will identify any quadratic form with the symmetric matrix implementing it.
In particular, for
\begin{align*}
	Q(k,l,m)
	=
		\begin{pmatrix}
			k \\
			l \\
			m
		\end{pmatrix}
		\begin{pmatrix}
			Q_{11} & Q_{12} & Q_{13} \\
			Q_{12} & Q_{22} & Q_{23} \\
			Q_{13} & Q_{23} & Q_{33}
		\end{pmatrix}
		\begin{pmatrix}
			k \\
			l \\
			m
		\end{pmatrix},
\end{align*}
the symbol $Q$ might refer to either the function or the matrix.
Likewise, the expressions $P(k,l)=k^2+l^2$ and $P=\Id$ are used interchangeably.
Often, we will fix $m$ and treat $Q_m(k,l):=Q(k,l,m)$ as a function of $k,l$.
For $m=0$,
\begin{align*}
	Q_0(k,l)
	=
	\begin{pmatrix}
		k \\
		l \\
	\end{pmatrix}
	\begin{pmatrix}
		Q_{11} & Q_{12} \\
		Q_{12} & Q_{22} \\
	\end{pmatrix}
	\begin{pmatrix}
		k \\
		l \\
	\end{pmatrix}
	.
\end{align*}
is again a quadratic form. %

On the qubit side, we will work in the $\mathfrak{so}(4)$-picture.
The orthogonality conditions will be verified separately for the span
$\K$ of $\{K_{-1}, K_0, K_1\}$
and the 
span
$\J$  of $\{J_{-1}, J_0, J_1\}$.

The lengthy proof is broken into a sequence of lemmas.

\begin{lemma}\label{lem:two_unitary_kj}
	An element 
	$U\in U\big( \CC^6 \otimes \CC^6 \big)$ 
	is two-unitary if and only if
	\begin{align}\label{eqn:two_unitarity_on_kj}
		U(\mathcal{L}^{(3)}\otimes \K) U^\dagger,
		\quad
		U(\mathcal{L}^{(3)}\otimes \J) U^\dagger,
		\quad\text{and}\quad
		U(\mathcal{L}^{(3)}\otimes \Id) U^\dagger
		\qquad\text{are orthogonal to}\qquad
		\mathcal{L}_0 \oplus \mathcal{R}_0.
	\end{align}
\end{lemma}

\begin{proof}
	This follows from Eq.~(\ref{eqn:two_unitarity_direct_sum}) and the fact that the span of $\K, \J, \Id$ equals the span of 
	$\mathcal{L}^{(2)}, \R^{(2)}$.
\end{proof}

We start by treating the first condition of Lem.~\ref{lem:two_unitary_kj}, regarding the space $\K$.

\begin{lemma}\label{lem:k_lemma}
	If
	$P=\Id$,
	then
	$U_\lambda (\mathcal{L}^{(3)}\otimes \J) U_\lambda^\dagger$ is orthogonal to 
	$\mathcal{L}_0 \oplus \mathcal{R}_0$
	if and only if
	\begin{align}\label{eqn:k_orth_cond}
		\omega^{r^2+s^2} \big(\F \omega^{Q_m}\big)(r,s) 
		+
		\overline{
			\omega^{r^2+s^2} \big(\F \omega^{Q_m}\big)(-r,-s)  
		} %
		&=0,
		\qquad
		r,s,m\in\ZZ_3^2.
	\end{align}
\end{lemma}

\begin{proof}
	Let $\vec a \in\ZZ_3^4$.
	We need a generalization of Eq.~(\ref{eqn:controlled_wm}):
	\begin{align*}
		U_\lambda
		\big(
			w(\vec a)
			\otimes |\Phi_m\rangle\langle\Phi_\emptyset|
		\big)
		U_\lambda^\dagger
		&=
		(U_{WH}\otimes \langle m|) \big(
			D_3 U_{WH}^\dagger (w(\vec a) \otimes \Id) U_{WH} (D_2^\dagger\otimes \Id)
		\big)
		(U_{WH}^\dagger\otimes|m\rangle)
			\otimes |\Phi_m\rangle\langle\Phi_\emptyset|.
	\end{align*}
	The factor acting on the qutrit system simplifies to
	\begin{align*}
		&
		(U_{WH}\otimes \langle m|) \big(
			D_3 U_{WH}^\dagger (w(\vec a) \otimes \Id) U_{WH} (D_2^\dagger\otimes \Id)
		\big)
		(U_{WH}^\dagger\otimes|m\rangle)
		\nonumber \\
		=&
		(U_{WH}\otimes \langle m|) 
			(D_3 D_2^\dagger\otimes \Id)
		(U_{WH}^\dagger\otimes|m\rangle) 
		w(S_2 \vec a)
		\nonumber \\
		=&
		\frac1{3}
		\sum_{r,s}
		(\F \omega^{Q_m})(r,s)
		(U_{WH}\otimes \langle m|) 
			Z^r\otimes Z^s
		(U_{WH}^\dagger\otimes|m\rangle) 
			w(S_2 \vec a)
		\nonumber 
			\\
		=&
		\frac1{3}
		\sum_{r,s}
		(\F \omega^{Q_m})(r,s)
		w(S_{WH} (r,s,0,0)) w(S_2 \vec a)
		,
	\end{align*}
	where $S_{WH}$, $S_2$ have been defined in in Eqs.~(\ref{eqn:swh}) and (\ref{eqn:23cliff}).
	If $\vec a = (p,0,q,0)$, so that $w(\vec a) \in \mathcal{L}^{(3)}$, then
	\begin{align*}
		S_{WH} 
		\begin{pmatrix}
			r\\
			s\\
			0\\
			0\\
		\end{pmatrix}
		+
		S_2 
		\begin{pmatrix}
			p\\
			0\\
			q\\
			0\\
		\end{pmatrix}
		=
		\begin{pmatrix}
			r\\
			-r\\
			-s\\
			-s\\
		\end{pmatrix}
		+
		\begin{pmatrix}
			p+q\\
			-q\\
			-p+q\\
			-p\\
		\end{pmatrix}
		=
		\begin{pmatrix}
			r+p+q\\
			-r-q\\
			-s-p+q\\
			-s-p
		\end{pmatrix}
		=: \vec l
		\,\Rightarrow\,
		w(S_{WH} (r,s,0,0)) w(S_2 \vec a)
		=
		\omega^{-ps-qr}    %
		w\big(\vec l\big).
	\end{align*}
	The image of $ w(p,q)\otimes\Id \otimes |\Phi_\emptyset\rangle\langle\Phi_m|$ can be computed in complete analogy.
	Altogether,
	\begin{align}%
		&
		U_\lambda
		\Big(
			w(p,q)\otimes\Id
			\otimes 
			K_m
		\Big)
		U_\lambda^\dagger
		=
		U_\lambda
		\Big(
			w(p,q)\otimes\Id
			\otimes 
			\big(
				|\Phi_\emptyset\rangle\langle\Phi_m|
				-
				|\Phi_m\rangle\langle\Phi_\emptyset|
			\big)
		\Big)
		U_\lambda^\dagger
		\nonumber \\
		=&
			\frac1{3}
			\sum_{r,s}
			w\big(\vec l(r,s,p,q)\big)
		\Big(
			(\F \omega^{Q_m})(r,s)
				\omega^{-ps-qr}
				\otimes 
				|\Phi_\emptyset\rangle\langle\Phi_m|
				-
				\overline{(\F \omega^{Q_m})(-r,-s) \omega^{-ps-qr}}
				\otimes 
				|\Phi_m\rangle\langle\Phi_\emptyset|
			\Big). \label{eqn:k_deloc_res}
	\end{align}

	Now compute the overlap between (\ref{eqn:k_deloc_res}) 
	and a basis $\{w(p_1,0,q_1,0), w(0,p_2,0,q_2)\} \otimes \{ K_i, J_j\}$ of 
	$\mathcal{L}_0 \oplus \mathcal{R}_0$.
	On the qutrit factor, $w(\vec l)$ is orthogonal to the local observables unless it itself lies in either $\mathcal{L}^{(3)}$ or $\R^{(3)}$ which happens if and only if
	\begin{align*}
		\begin{pmatrix}
			p\\
			q
		\end{pmatrix}
		=
		\begin{pmatrix}
			0&-1\\
			-1&0
		\end{pmatrix}
		\begin{pmatrix}
			r\\
			s
		\end{pmatrix}
		\qquad
		\text{or}
		\qquad
		\begin{pmatrix}
			p\\
			q
		\end{pmatrix}
		=
		\begin{pmatrix}
			1&1\\
			1&-1
		\end{pmatrix}
		\begin{pmatrix}
			r\\
			s
		\end{pmatrix}.
	\end{align*}
	In both cases, $\omega^{-ps-qr}=\omega^{r^2+s^2}$.
	On the qubit factor, (\ref{eqn:k_deloc_res}) is manifestly orthogonal to all basis elements except possibly $K_m$.
	Demanding that the overlap with $K_m$, too, vanishes gives the advertised condition.
\end{proof}

For both artisanal solutions, the form $P$ has rank one.
This is necessarily the case, at least for its restriction to $\ZZ_3^2$.

\begin{lemma}\label{eqn:rank_lemma}
	Equation (\ref{eqn:k_orth_cond}) implies that $Q_0$ has rank one.
\end{lemma}

\begin{proof}
	In the case $r=s=0$, Eq.~(\ref{eqn:k_orth_cond}) is equivalent to $\sum_{kl} \omega^{Q_0(k,l)}$ being pure imaginary.
	Expressing the sum in a basis of $\ZZ_3^2$ that is orthogonal with respect to the bilinear form $Q_0$, it is seen to be a product of two Gauss sums.
	By Eq.~(\ref{eqn:gauss_sum_3}), if the product is pure imaginary, then necessarily the quadratic coefficient has to be non-zero for exactly one of the two factors.
\end{proof}

The next lemma clarifies the orbits of rank-one symmetric matrices in $\ZZ_3^{2\times 2}$.
We limit ourselves to operations that preserve the choice $P=\Id$ already made.
To this end,
let $O(\ZZ_3^2)$ be the group of matrices that are orthogonal in the sense that $G G^t=\Id$.

\begin{lemma}\label{lem:orbits}
	Under $O(\ZZ_3^2)$ acting as $S \mapsto G S G^t$, there are exactly four orbits of rank-one $2\times 2$-matrices.
	Representatives are
	\begin{align*}
		R\sparse_a
		:=
		\begin{pmatrix}
			0 & 0 \\
			0 & a
		\end{pmatrix},
		\qquad
		R\sym_a:=
		-
		\begin{pmatrix}
			a & a \\
			a & a
		\end{pmatrix},
		\qquad
		a\in\{\pm1\}.
	\end{align*}
\end{lemma}

\begin{proof}
	Let $S$ be a symmetric rank-1 matrix on $\ZZ_3^2$.

	If the second column of $S$ is zero, substitute $S \mapsto X S X^t$.
	(Here, $X$ is the Pauli matrix, with entries interpreted as elements of $\ZZ_3$).
	Clearly, $X\in O(\ZZ_3^2)$, and $XSX^t$ has a non-zero second column.

	Then the lower-right entry is non-zero, for else the matrix would have rank $2$.
	If the upper-right entry is zero, we have attained the sparse normal form.
	If the upper-right entry is $2$, substitute $S \mapsto ZSZ^t$ to make it $1$.
	Then $S$ can be mapped to $R\sym_a$ by
	\begin{align*}
		\begin{pmatrix}
			1 & 0 \\
			0 & a
		\end{pmatrix}
		\begin{pmatrix}
			-a & 1 \\
			1 & -a
		\end{pmatrix}
		\begin{pmatrix}
			1 & 0 \\
			0 & a
		\end{pmatrix}
		=
		-
		\begin{pmatrix}
			a & a \\
			a & a
		\end{pmatrix}.
	\end{align*}

	The normal forms lie on distinct orbits.
	A map that sends one of the left matrices to one of the right matrices must map 
	$\vec e_2$ to 
	$ (\vec e_1+\vec e_2)$
	or to
	$-(\vec e_1+\vec e_2)$.
	Because
	\begin{align*}
		\vec e_2^t \vec e_2
		=
		1
		\neq
		2
		=
		\big(\pm (\vec e_1+\vec e_2)\big)^t
		\big(\pm (\vec e_1+\vec e_2)\big),
	\end{align*}
	such a map cannot be orthogonal.
	The discriminant of $R\sparse_a$ is $a$.
	As the discriminant is invariant under any change of basis, it is not possible to  map $R\sparse_1$ to $R\sparse_{-1}$.
	The symmetric case is treated analogously.
\end{proof}

The relevance of this classification for the problem at hand stems from the next lemma, which says that the $\ZZ_3^2$ basis changes of Lem.~\ref{lem:orbits} are implemented as symmetries on doubly perfect functions.

\begin{lemma}\label{lem:relating_symmetries}
	For $G\in\GL(\ZZ_3^2)$, define
	$\hat G = 4 G + 3\Id \in\ZZ_6^{2\times 2}$.
	Then $\hat G^{-t} \in \ESp(\ZZ_6^2)$. 
	The action 
	$\lambda\mapsto \lambda \circ \hat G^t$
	corresponds to 
	$P \mapsto GPG^t$, $Q\mapsto (G\oplus 1)Q(G\oplus 1)^t$ 
	and preserves doubly perfect functions.
\end{lemma}

\begin{proof}
	Using $\det G\in \{\pm1\}$, the following identity holds modulo six:
	\begin{align*}
		\det \hat G = 4^2 \det G + 3^2 \det \Id = \det G
	\end{align*}
	Hence also $\det \hat G^{-t}\in \{\pm1\}$.
	By Sec.~\ref{sec:list_of_symmetries}, Item~\ref{itm:extended}, $\ESp(\ZZ_6^2)=\GL(\ZZ_6^2)$, which proves the first claim.

	Using the notions of Eq.~(\ref{eqn:decompose}), for $\vec a, \vec a' \in \ZZ_6^2$,
	\begin{align*}
		\vec a' = \hat G^t \vec a
		\qquad
		\Leftrightarrow
		\qquad
		\begin{pmatrix}
			k' \\
			l'
		\end{pmatrix}
		=
		G
		\begin{pmatrix}
			k \\
			l
		\end{pmatrix}
		\quad\text{and}\quad
		\begin{pmatrix}
			x' \\
			y'
		\end{pmatrix}
		=
		\Id
		\begin{pmatrix}
			x \\
			y
		\end{pmatrix}.
	\end{align*}
	Thus $\lambda(\hat G^{t}\,\vec a) = \omega^{\phi(k',l';m)}$. 
	But
	\begin{align*}
		Q(k',l') = 
		(k,l)
		G Q G^{t} 
		(k,l)^t,
		\qquad
		P(k',l',m) = 
		(k,l,m)
		(G\oplus 1) Q (G \oplus 1)^t
		(k,l,m)^t.
	\end{align*}
	Finally, by Sec.~\ref{sec:list_of_symmetries}, $\lambda\mapsto (\hat G^{-t}\lambda) = \lambda\circ \hat G^t$ is a symmetry of doubly perfect functions.
\end{proof}

While Lem.~\ref{lem:orbits} says that basis changes on $\ZZ_3^2$ cannot map functions with 
$P=\Id, Q_0=R\sparse_a$
to functions with 
$P=\Id, Q_0=R\sym_a$,
we now show that this \emph{can} be achieved if we also allow for a global change of sign $\lambda\mapsto-\lambda$.
This will allow us to reduce the ``symmetric'' case to the ``sparse'' case.

\begin{lemma}\label{lem:orbit_connector}
	There exists a $G\in\GL(\ZZ_3^2)$ such that:
	\begin{enumerate}
		\item
			If $\lambda$ is defined by $P=\Id$ and a form $Q$ with $Q_0$ in the $O(\ZZ_3^2)$-orbit of $R\sym_a$.
			Then $-\lambda\circ G^t$ lies in the $O(\ZZ_3^2)$-orbit of a function $\lambda'$, with
			$P'=\Id$ and $Q'=R\sparse_a$.

		\item
			Specifically, $-\lambda\sym \circ G^t = \lambda\sparse$.
	\end{enumerate}
\end{lemma}

\begin{proof}
	For the choice
	\begin{align*} %
		G:=
		\begin{pmatrix}
			1  & 2 \\
			2  & 2
		\end{pmatrix}
		\qquad\text{direct calculations give}\qquad
		G (-\Id)   G^t = \Id,
		\quad
		G (-R\sym_a) G^t = R\sparse_a,
		\quad
		-\lambda\sym \circ G^t = \lambda\sparse.
	\end{align*}
	The matrix $G$ also normalizes the orthogonal maps in the sense that
	\begin{align*}
		O \in O(\ZZ_3^2)
		\quad\Rightarrow\quad
		(G O G^{-1}) (G O G^{-1})^t
		=
		G O (-\Id) O^t G^t
		=
		G (-\Id) G^t
		=
		\Id
		\quad\Rightarrow\quad
		(G O G^{-1}) \in O(\ZZ_3^2).
	\end{align*}
	Now let $\lambda$ such that $P=\Id$ and there exists an $O_1\in O(\ZZ_3^2)$ such that $O_1 Q_0 O_1^t = R\sym_a$.
	Set $O_2=GO_1G^{-1} \in O(\ZZ_2^3)$.
	Then for $-\lambda\circ G^t \circ O_2^t$, we have
	\begin{align*}
		P' = O_2G(-\Id) G^tO_2^t = \Id,
		\qquad
		Q_0' 
		= O_2 G(O_1^{-1} (-R\sym_a) O_1^{-t}) G^t O_2^t
		= G (-R\sym_a) G^t
		= R\sparse_a.
	\end{align*}
\end{proof}

With the previous lemma in mind, we now concentrate on the case where $Q$ is of the form
\begin{align}\label{eqn:q2}
	Q=
	\begin{pmatrix}
		0 & 0 & Q_{13} \\
		0 & a & Q_{23} \\
		Q_{13} & Q_{23} & Q_{33}
	\end{pmatrix}.
\end{align}

\begin{lemma}\label{lem:discriminant}
	Assuming the form (\ref{eqn:q2}),
	Eq.~(\ref{eqn:k_orth_cond}) holds if and only if 
	\begin{align*}
		Q_{13}=0, \quad
		a = 1, \quad
		Q_{33}=-aQ_{23}^2.
	\end{align*}
\end{lemma}

\begin{proof}
	Under the assumption, the first summand in Eq.~(\ref{eqn:k_orth_cond}) becomes
	\begin{align*}
		\omega^{r^2+s^2}
		\frac13
		\sum_{kl}
		\omega^{a l^2 - Q_{13} k m - Q_{23} l m + Q_{33} m^2 - rk - sl} 
		&=
		a i \sqrt 3
		\omega^{
			r^2+
			(1-a)s^2+
			(Q_{33} -a Q_{23}^2)m^2 
			+a Q_{23} m s
		}
		\delta(Q_{13} m + r)
		.
	\end{align*}
	Thus, Eq.~(\ref{eqn:k_orth_cond}) states that the following sum must evaluate to zero for all $r,s,m$:
	\begin{align*}
			\omega^{r^2+s^2} \big(\F \omega^{Q_m}\big)(r,s) 
			+
			\overline{
				\omega^{r^2+s^2} \big(\F \omega^{Q_m}\big)(-r,-s)  
			} 
		=&+ a i
			\sqrt 3
			\omega^{
				r^2+
				(1-a)s^2+
				(Q_{33} -a Q_{23}^2)m^2 
				+a Q_{23} m  s
			}
			\delta(Q_{13} m + r)\\
			&-a i
			\sqrt 3
			\omega^{
				-r^2
				-(1-a)s^2
				-(Q_{33}-a Q_{23}^2)m^2 
				-a Q_{23} m s
			}
			\delta(Q_{13} m - r).
		\end{align*}
		Because the exponentials are non-zero,
		the sum can only vanish identically if the two delta terms realize the same function of $r$, which happens if and only if $Q_{13}=0$.
		This reduces the condition to the $r=0$-case.
		Equating the exponents gives
		\begin{align*}
				(1-a)s^2
				+(Q_{33} -a Q_{23}^2)m^2 
			=0.
	\end{align*}
	Because $s, m$ are arbitrary, this is equivalent to $a=1$ and $Q_{33}=-aQ_{23}^2$.
\end{proof}

In other words, $Q$ of the form (\ref{eqn:q2}) defines a doubly perfect function if and only if
\begin{align}\label{eqn:q3}
	Q=
	\begin{pmatrix}
		0 & 0 & 0 \\
		0 & 1 & Q_{23} \\
		0 & Q_{23} & Q_{23}^2
	\end{pmatrix}.
\end{align}
To make further progress, we now turn to the second condition of Lem.~\ref{lem:two_unitary_kj}, regarding the space $\J$.

\begin{lemma}\label{lem:j_lemma}
	If
	$P=\Id$
	and $Q$ is of the form (\ref{eqn:q3}), then
	$U_\lambda (\mathcal{L}^{(3)}\otimes \J) U_\lambda^\dagger$ is orthogonal to 
	$\mathcal{L}_0 \oplus \mathcal{R}_0$
	if and only if
	$Q_{23}\neq 0$.
\end{lemma}

\begin{proof}
	We proceed as in the proof of Lem.~\ref{lem:k_lemma}:
	\begin{align*}
		U_\lambda
		\big(
			w(\vec a)
			\otimes |\Phi_m\rangle\langle\Phi_n|
		\big)
		U_\lambda^\dagger
		&=
		(U_{WH}\otimes \langle m|) \big(
			D_3 U_{WH}^\dagger (w(\vec a) \otimes \Id) U_{WH} (D_2^\dagger\otimes \Id)
		\big)
		(U_{WH}^\dagger\otimes|m\rangle)
			\otimes |\Phi_m\rangle\langle\Phi_n|
			.
	\end{align*}
	The factor acting on the qutrit system is
	\begin{align*}
		U_{WH} D_{3,m} 
		U_{WH}^\dagger (w(\vec a) \otimes \Id) 
		U_{WH} D_{3,n}^\dagger U_{WH}^\dagger
		=&
		U_{WH} D_{3,m} D_{3,n}^\dagger U_{WH}^\dagger
		w(S_{3}^{UL}\vec a)
		\\
		=&
		\omega^{Q_{33}^2 (m^2-n^2)}
		U_{WH} 
		(\Id\otimes Z^{Q_{23}(m-n)})
		U_{WH}^\dagger
		\omega^{[\vec a, (0, Q_{23},0,0)]}
		w((S_{Z}^Q)_0\vec a)
		\\
		=&
		\omega^{Q_{33}^2 (m^2-n^2) + [\vec a, 0, Q_{23},0,0)]}
		w\big(S_{WH} ( 0, Q_{23}(m-n),0,0)\big) 
		w((S_{Z}^Q)_0\vec a)
		,
	\end{align*}
	where 
	$S^Q_Z$ is the symplectic matrix associated with the quadratic form 
	$N=\diag(1,-1)$ 
	as in Eq.~(\ref{eqn:phase_gates}).
	For $\vec a = (p,0,q,0)$,
	\begin{align*}
		S_{WH} 
		\begin{pmatrix}
			0\\
			Q_{23}(n-m)\\
			0\\
			0\\
		\end{pmatrix}
		+
		S^Q_Z
		\begin{pmatrix}
			p\\
			0\\
			q\\
			0\\
		\end{pmatrix}
		=
		\begin{pmatrix}
			p+q\\
			-q\\
			p+q+Q_{23}(m-n)\\
			p + Q_{23}(m-n)
		\end{pmatrix}
		=: \vec l.
	\end{align*}
	If $Q_{23}\neq 0$, then the operator $w(\vec l)$ is 
	orthogonal to $\mathcal{L}^{(3)}_0\oplus\R^{(3)}_0$ 
	for all $p,q$.
	Conversely, if $Q_{23}=0$, then
	\begin{align*}
		U_\lambda 
			\big(w(0) \otimes (|\Phi_m\rangle\langle\Phi_n| - |\Phi_n\rangle\langle\Phi_m|)\big)
		U_\lambda^\dagger
		=
		w(0)\otimes \big(
			\omega^{Q_{33}^2 (m^2-n^2) } 
			|\Phi_m\rangle\langle\Phi_n| 
			- 
			\omega^{-Q_{33}^2 (m^2-n^2) } 
			|\Phi_n\rangle\langle\Phi_m|
		\big),
	\end{align*}
	which has overlap $2\Re \omega^{Q_{33}^2 (m^2-n^2) }\neq 0$ with 
	$\Id\otimes ( |\Phi_m\rangle\langle\Phi_n| - |\Phi_n\rangle\langle\Phi_m| ) \in \mathcal{L}_0\oplus\mathcal{R}_0$.
\end{proof}

Hence, a function $\lambda$ with $P=\Id$ and $Q_0$ 
of the form (\ref{eqn:q2}) is doubly perfect if and only if $Q$ is one of
\begin{align}\label{eqn:q4}
	Q\sparse_+=
	\begin{pmatrix}
		0 & 0 & 0 \\
		0 & 1 & 1 \\
		0 & 1 & 1
	\end{pmatrix},
	\qquad
	Q\sparse_-=
	\begin{pmatrix}
		0 & 0 & 0 \\
		0 & 1 & 2 \\
		0 & 2 & 1
	\end{pmatrix}
	.
\end{align}
The two forms are equivalent under $O(\ZZ_3^2)$, with $-\Id$ applied to the first two coordinates switching between them.
Rewriting the matrices as quadratic forms, one immediately sees that $Q=Q\sparse_+$ leads to $\lambda=\lambda\sparse$.

At this point, we can establish the converse direction of the main theorem:
Namely if $P$ and $Q$ define a doubly perfect function $\lambda$, then $\lambda$ is equivalent, under $\GL(\ZZ_3^2)$, to either $\lambda\sym$ or $\lambda\sparse$.

\begin{proof}[Proof (of the converse part of Thm.~\ref{thm:main})]
	Assume that $P$ and $Q$ define a doubly perfect function $\lambda$. 
	Up to the action of $\GL(\ZZ_3^2)$, we may assume that $P=\Id$ and $Q_0$ is either $R\sparse_a$ or $R\sym_a$ (Sec.~\ref{sec:u2} and Lem.~\ref{lem:orbits}).

	In the first case, 
	$\lambda$ can be mapped to $\lambda\sparse$ by an element of $O(\ZZ_3^2)$ (Lem.~\ref{lem:discriminant}).

	In the second case, 
	combining Lem.~\ref{lem:orbit_connector} with the previous case, 
	there is an $O\in O(\ZZ_3^2)$ such that
	$-\lambda\circ G^t \circ O^t =\lambda\sparse$.
	But then, again by Lem.~\ref{lem:orbit_connector},
	it holds that
	$-\lambda\circ G^t \circ O^t \circ G^{-t} =\lambda\sym$.
\end{proof}

It remains to prove the direct part, i.e.\ that these orbits actually consist of doubly perfect functions.
To establish this, we must verify the third and final condition of Lem.~\ref{lem:two_unitary_kj}.
By the discussion below Eq.~(\ref{eqn:q4}), it suffices to treat the sparse artisanal solution $P=\Id, Q=Q\sparse_+$.

\begin{lemma}\label{lem:sparse_works}
	For $\lambda=\lambda\sparse$, it holds that
		$U_\lambda(\mathcal{L}^{(3)}\otimes \Id) U^\dagger_\lambda$
		is orthogonal to
		$\mathcal{L}_0 \oplus \mathcal{R}_0$.
\end{lemma}

\begin{proof}
	A basis for $U_\lambda(\mathcal{L}^{(3)}\otimes \Id) U^\dagger_\lambda$ appears in Sec.~\ref{sec:qutrit_algebras}.
	The projection of these elements onto $\mathcal{L}^{(3)}_0\oplus\R^{(3)}_0$ vanish, except when $q=-p$, in which case the projection is 
	\begin{align*}
		\Id\otimes w(p,p) \otimes V_t w(p,0) V_t^\dagger
		=
		\Id\otimes w(p,p) \otimes \Big(\sum_{m} \omega^m |\Phi_m\rangle\langle\Phi_m|\Big).
	\end{align*}
	But all elements of $\mathcal{L}^{(2)}_0\oplus\mathcal{R}^{(2)}_0$ are off-diagonal, and thus orthogonal to the sum in parentheses.
\end{proof}

This concludes the proof of Thm.~\ref{thm:main}.

\subsubsection{Further properties of the orbits}

We close with two minor remarks.

The unitaries $U_{\lambda\sym}$ and $U_{\lambda\sparse}$  
are not unitarily equivalent, not even under global non-Clifford unitaries. 
This follows from:

\begin{lemma}
	It holds that $\tr U_{\lambda\sym} \neq \tr U_{\lambda\sparse}$.
\end{lemma}

\begin{proof}
	Let $R$ be a quadratic form on a $\ZZ_d^n$.
	Then, evaluating the sum in an orthogonal basis of the form and using Sec.~\ref{sec:gauss}, 
	\begin{align*}
		\sum_{\vec a\in\ZZ_d^n} \omega_d^{\vec a R \vec a} = 
		\ell\, \gamma_d^r\, d^{n-r},
	\end{align*}
	where $\ell$ is the discriminant of the form and $r$ its rank.
	For $\lambda\sparse$, 
	\begin{align*}
		\det( P+Q)
		=
		\det
		\begin{pmatrix}
			1 & 0 & 0 \\
			0 & 2 & 1 \\
			0 & 1 & 1
		\end{pmatrix}
		=
		1
	\end{align*}
	so that $P+Q$ has rank $3$ and discriminant $1$.
	Likewise, $P$ has rank $2$ and discriminant $1$.
	Hence
	\begin{align*}
		\tr U_{\lambda\sparse} 
		&=
		\sum_{\vec a\in\ZZ_3^2} \omega^{ \vec a P \vec a}
		+
		\sum_{\vec a\in\ZZ_3^3} \omega^{ \vec a(P+Q) \vec a}
		=
		\gamma^2 + \gamma^3 = -3(1+ i \sqrt 3).
	\end{align*}
	Since the trace is not real, 
	$\tr U_{\lambda\sym} = \overline{\tr U_{\lambda\sparse}} \neq \tr U_{\lambda\sparse}$.
\end{proof}

Having identified the geometric structure behind the solutions, it is a simple exercise to count them.

\begin{lemma}
	There are $24$ doubly perfect functions in each orbit.
\end{lemma}

\begin{proof}
	We claim that $\SL(\ZZ_3^2)=\Sp(\ZZ_3^2)$ acts freely on $\lambda\sparse$ and generates its $\GL(\ZZ_3^2)$ orbit.

	For the first part:
	The stabilizer subgroup of $P$ in $\SL(\ZZ_3^2)$ is $\SO(\ZZ_3^2)$.
	The stabilizer subgroup of $R\sparse$ in $\SO(\ZZ_3^2)$ is $\{\Id, -\Id\}$.
	But $-\Id$ toggles the sign in $Q\sparse_\pm$ (c.f.~Eq.~(\ref{eqn:q4})),
	implying that the stabilizer subgroup of $\lambda\sparse$ in $\SL(\ZZ_3^2)$ is trivial.
	For the second part:
	$\GL(\ZZ_3^2)=\SL(\ZZ_3^2) \rtimes \{\Id, Z\}$, and $Z$ stabilizes $\lambda\sparse$.

	Hence there are $|\SL(\ZZ_3^2)|=24$ elements in the orbit of $\lambda\sparse$.
	The two orbits obviously contain the same number of elements.
\end{proof}

\section{Summary and Outlook}

In this work, we have developed new methods for the analytic construction of two-unitaries.
While we did succeed in finding a human-checkable solution to the problem of constructing a two-unitary of order six,
the process of actually checking it isn't all too pleasant (ask us how we know).
It would be nice to have an abstract argument that avoids the more tedious calculations.
Some possible directions we have not been able to exploit, and that aren't mentioned elsewhere in the text, are these:
\begin{itemize}
	\item
		Maybe one can realize a two-unitary as the restriction of a simpler object in a larger-dimensional space
		(the numerical coincidences $6=\dim\Sym^2(\CC^3)=\dim \wedge^2(\CC^4)$ have not escaped our attention).
		For example, assume that there is a matrix group $G$ acting on $\CC^d\otimes\CC^d$, and affording a ``transversal projection $P$'' in its commutant, in the sense that $[G, P\otimes P]=0$.
		Then if $G$ contains a two-unitary (of order $d$), it restricts to a two-unitary of order $\tr P$.
		Intriguingly, representations of the Clifford group as symmetric tensor powers do sometimes have such transversal projections in their commutant \cite{gross2021schur,montealegre2021rank,montealegre2022duality}.
		Unfortunately, we have not found a way to exploit this structure for the problem treated here.

	\item
		While mixing $\ZZ_2$ and $\ZZ_3$ arithmetic leads to manifestly non-linear functions, 
		one might reformulate them in terms of low-order polynomials.
		For example, in the context of Eq.~(\ref{eqn:Delta}),
		we frequently need to compute differences of the form
		\begin{align*}
			(\widehat{x+u}) - (\widehat{y+v}),
			\quad\text{which can be expressed as}\quad
			(\hat x - \hat y)
			+
			(\hat u - \hat v)
			-xu + yv.
		\end{align*}
		This feels similar to the way the WH group over $\ZZ_2^n$ mixes $\ZZ_2$ and $\ZZ_4$ arithmetic.
		It may be hoped that a simpler analysis can be based on this observation.

	\item
		The two-unitaries $U_\lambda$ are \emph{semi-Clifford} \cite{zeng2008semi}.
		Maybe their properties can be related to the Clifford hierarchy
		\cite{gottesman1999demonstrating},
		along the lines of works such as \cite{beigi2008c3, rengaswamy2019unifying,chen2024characterising}?
\end{itemize}

\subsection{Acknowledgments}

We thank 
Wojciech Bruzda,
Hans Georg Feichtinger,
Markus Grassl,
Ren\'e Schwonnek,
Andreas Winter, 
Nikolai Wyderka,
and
Karol {\.Z}yczkowski
for discussions.
We are particularly grateful to Suhail Ahmad Rather
for re-triggering our interest in the problem,
and the organizers of the
\emph{6th Workshop on Algebraic Structures in Quantum Computations}
for providing the venue for this exchange.

This work was supported by
Germany's Excellence Strategy -- Cluster of Excellence Matter and Light for Quantum Computing (ML4Q), EXC 2004/1 (390534769).

\bibliographystyle{plain}
\bibliography{bibliography}

\appendix

\section{Auxiliary Calculations}

\subsection{Quadratic Gauss sums}
\label{sec:gauss}

We recall a standard calculation involving quadratic Gauss sums.
For $d$ an odd prime and $a\neq 0,b,c\in\ZZ_d$,
completing the square gives
\begin{align*}
	\sum_{x\in\ZZ_d} \omega_d^{a x^2+bx+c} 
	&=
	\omega_d^{c}
	\sum_x \omega_d^{a(x^2+b/a x)} 
	=
	\omega_d^{-a\left(\frac{b}{2a}\right)^2+c}
	\sum_x \omega_d^{a\big(x+\frac{b}{2a}\big)^2} 
	=
	\omega_d^{-a\left(\frac{b}{2a}\right)^2+c}
	\,
	\left(\frac{a}d\right)
	\,
	\gamma_d,
	\qquad
	\gamma_d
	:=
	\sum_{x} \omega_d^{x^2},
\end{align*}
where all divisions are to be performed modulo $d$ and where 
$\left(\frac{a}d\right)$
is the Legendre symbol, which is $1$ if $a$ is a square modulo $d$ and $-1$ else.
Here, we are mainly interested in the case $d=3$, where the 
above implies
\begin{align}
	\begin{split} \label{eqn:gauss_sum_3}
		\sum_{x\in\ZZ_3} \omega^{a x^2+bx+c} 
		&=
		\gamma
		a
		\omega^{-ab^2+c},
		\quad
		\gamma
		=
		\sqrt 3 i
		\qquad(a\neq 0), \\
		\sum_{x\in\ZZ_3} \omega^{bx+c} 
		&=
		3 \omega^c \delta(b).
	\end{split}
\end{align}

\end{document}